\documentclass[10pt,twocolumn,twoside]{IEEEtran}

\IEEEoverridecommandlockouts                              
\usepackage{graphics} 
\usepackage{epsfig} 
\usepackage{mathptmx} 
\usepackage{times} 
\usepackage{amsmath} 
\usepackage{amssymb}  
\usepackage{dsfont}
\usepackage{subfigure}
\usepackage{cancel}
\usepackage{margins}
\usepackage{mathtools}
\usepackage{multirow}  
\newtheorem{thm}{Theorem}[section]

\newtheorem{cor}{Corollary}[section]
\newtheorem{lem}{Lemma}[section]
\newtheorem{remark}{Remark}[section]
\newtheorem{assumption}{Assumption}[section]
\newtheorem{example}{Example}[section]

\title{Phase Balancing of Two and Three-Agent Heterogeneous Gain Systems With Extensions to Multiple Agents}

\author{Anoop Jain and Debasish Ghose 
\thanks{A. Jain is a graduate student at the Guidance, Control and Decision System Laboratory (GCDSL) in the Department of
                 Aerospace Engineering, Indian Institute of Science,
                  Bangalore, India (email: anoopj@aero.iisc.ernet.in).}
\thanks{D. Ghose is a Professor at the Guidance, Control and Decision System Laboratory (GCDSL) in the Department of
              Aerospace Engineering, Indian Institute of
              Science, Bangalore, India (email: dghose@aero.iisc.ernet.in).}
\thanks {This work is partially supported by Asian Office of Aerospace Research and Development
(AOARD).}}

\begin{document}
\maketitle

\begin{abstract}
This paper studies the phase balancing of a two and three-agent system where the agents are coupled through heterogeneous controller gains. Balancing refers to the situation in which the movement of agents causes the position of their centroid to become stationary. We generalize existing results and show that by using heterogeneous controller gains, the velocity directions of the agents in balanced formation can be controlled. The effect of heterogeneous gains on the reachable set of these velocity directions is further analyzed. For the two-agent's case, the locus of steady-state location of the centroid is also analyzed against the variations in the heterogeneous controller gains. Simulations are given to illustrate the theoretical findings.
\end{abstract}


\section{Introduction}
\subsection{Prelude}
Multi-agent systems exhibit different collective behaviors because of their potential applications in several areas such as formation control of unmanned aerial vehicles (UAVs) \cite{Beard2008}, \cite{Mesbhai2010}, autonomous underwater vehicles (AUVs) \cite{Sepulchre2007} and spacecraft \cite{Beard2004}, cooperative robotics \cite{Zqu2009}, and sensor networks \cite{cortes2004}. In this paper, our main interest is to study a particular type of collective formation of a multi-agent system so called balanced formation. Balancing refers to the situation when all the agents of a group move in such a way that their position centroid remains stationary. A contrary notion of balancing is synchronization, which refers to the situation when all the agents of a group have a common velocity direction. The phenomenon of synchronization is widely studied in the literature, for instance refer \cite{Strogatz2000}$-$\cite{Jain2016} and the references therein. In this paper, the phrases ``balanced formation" and ``phase balancing" are used interchangeably.

Recently, the important insights in understanding the phenomenon of phase synchronization and balancing have come from the study of the Kuramoto model \cite{Strogatz2000}, \cite{Jadbabaie2004}. This model is widely studied in the literature in the context of achieving synchronization and balancing in multi-agent systems. For instance in \cite{Sepulchre2007}, Kuramoto model type steering control law is derived to stabilize synchronized and balanced formations in a group of agents. The proposed control law in \cite{Sepulchre2007} operates with homogeneous controller gains, which gives rise to a balanced formation of agents with their unique phase arrangement. Recently, the effect of heterogeneity in various aspects have been studied in the literature. For example, \cite{Seyboth2014} considers heterogeneous velocities of the agents. In a similar spirit, in this paper, we consider that the controller gains are heterogeneously distributed, that is, they are not necessarily the same for each agent, and can be deterministically varied. It will be shown that this type of heterogeneity in the controller gains also leads to a balanced formation, in which a desired phase arrangement of the agents can be obtained by a proper selection of heterogeneous gains.

\subsection{Motivations}
The motivation to study balancing under heterogeneous controller gains is twofold \cite{Jain2016}. First, in many engineering applications in the field of aerial and underwater vehicles, it is required that all the vehicles move in a formation. Utilizing heterogeneity in the controller gains, the formation of these vehicles can be made to move in a desired direction, thus helping to explore an area of interest. Secondly, while implementing the control law physically for the homogeneous gains case, it is impossible to get identical controller gain for each agent. Thus, some errors in the individual controller gains is inevitable, leading to heterogeneity in the controller gains. It would be useful to know the effect of this heterogeneity on phase balancing performance of the multi-agent system.

\subsection{Literature Review}
The literature, related to achieving phase balancing or balanced formation in a multi-agent system, has focused on the controller design methodology which operates with homogeneous controller gains. In \cite{Sepulchre2007}, steering control laws are proposed to stabilize synchronized and balanced formations of a group of agents moving at unit speed. By taking into account the non-identical velocities of the agents, the stabilization of balanced formation is further discussed in \cite{Seyboth2014}. In \cite{Xu2013}, a modified Kuramoto model based control algorithm is proposed for making multiple agents spread out equidistantly on a circle, which is usually called splay formation. The splay phase arrangement \cite{Paley2005} is a special case of balanced formation, in which the phases are separated by multiples of $2\pi/N$ ($N$ being the number of agents). In \cite{Klein2008}, the asymptotic stability of the balanced set is proved in discrete time. An algorithm to stabilize synchronization and balancing in phase models on the $N$-torus, is proposed in \cite{Scardovi2007}. The clustering phenomenon, which referred to as the coexistence of synchronization and phase balancing, is also studied in the literature \cite{Okuda1993}$-$\cite{AJain2016} for coupled oscillators as well as multi-agent systems. Other than these, the phenomenon of phase balancing in coupled oscillators is discussed in the literature with various names like incoherent states, chimera sates, etc. \cite{Strogatz2000}, \cite{Strogatz1991}. Moreover, in the context of heterogeneity in the control gains, in \cite{Sinha2006}, heterogeneous controller gains have been used in a cyclic pursuit framework to obtain desired meeting points (rendezvous) and directions. The idea of dynamically adjustable control gains have been used in \cite{Ding2012} to study the pursuit formation of multiple autonomous agents.

\subsection{Contributions}
Many applications (like underwater exploration) are done with small number of vehicles as these are large and expensive. The strategies presented in this paper, although computationally feasible for more than three agents, are analytically tractable only for two and three agents. However, the analytical results provide important insights into the control of multi-agent systems as demonstrated by using simulations.

The contributions of the present paper are the following:
\begin{itemize}
\item A steering control law, which operates with heterogeneous controller gains, is proposed to asymptotically stabilize balanced formation of a group of $N$ agents.
\item It is proved analytically for the two and three-agent systems that the heterogeneity in the controller gains leads to a balanced formation, in which the desired arrangement of  agents' velocity vectors, can be obtained by a proper selection of the heterogeneous control gains.
\item The reachable set of the velocity directions of the agents in balanced formation is further analyzed under the effect of heterogenous control gains as well as homogeneous control gains with an inevitable non-uniform error.
\item Unlike all control gains being positive for phase balancing, it is analytically shown for the two-agent system that there exist a combination of both positive and negative values of the control gains which results in the further expansion of the reachable set of the agents' velocity directions in balanced formation.
\item For the two-agent system, we obtain the closed form expressions of the velocity directions, and analyze the locus of steady-state location of the centroid against the variations in the heterogeneous controller gains. Moreover, it is shown how this locus is useful in achieving balanced formation of agents about a desired steady-state location of the centroid.
\end{itemize}

The preliminaries of the present work have been presented in \cite{Jain2013}.

\subsection{Organization}
The paper is organized as follows: Section II describes the dynamics of the system and formulates the problem. In Section III, we analyze the effect of heterogeneous control gains on the velocity directions of agents in balanced formation. In Section IV, by deriving a less restrictive condition on the heterogeneous gains for the two-agent's case, we show that the reachable set of the velocity directions in balanced formation further expands. Section IV obtains the explicit expressions of the velocity directions of the two agents and their convergence point, and find its locus against the variations in the heterogeneous controller gains. Simulation results are provided in Section V. Finally, Section VI concludes the paper with a brief summary of future challenges.

\section{System Description and Problem Formulation}

\subsection{System Model}
A multi-agent system composed of $N$ autonomous agents, moving in a planar space, each assumed to have unit mass and unit speed, is considered in this paper and represented as
\begin{subequations}\label{modelNew}
\begin{align}
\label{modelNew1}\dot{r}_k & = e^{i\theta_k}\\
\label{modelNew2} \dot{\theta}_k & =  u_k; ~~~ k = 1, \ldots, N,
\end{align}
\end{subequations}
where, $r_k = x_k + i y_k \in \mathbb{C}$, $\dot{r}_k = e^{i\theta_k} = \cos\theta_k + i\sin\theta_k \in \mathbb{C}$ and $\theta_k \in \mathbb{S}^1$ are, respectively, the position, velocity and the heading angle of the $k^\text{th}$ agent, and $i = \sqrt{-1}$ denotes the standard complex number. The orientation, $\theta_k$ of the (unit) velocity vector represents a point on the unit circle $\mathbb{S}^1$, and is also referred to as the phase of the $k^\text{th}$ agent \cite{Strogatz2000}. The feedback control law $u_k \in \mathbb{R}$ controls the angular rate of the $k^\text{th}$ agent. If, $\forall k$, the control input $u_k$ is identically zero, then each agent travels at constant unit speed in a straight line in its initial direction $\theta_k(0)$ and its motion is decoupled from other agents. If, $\forall k$, the control input $u_k = \omega_0$ is constant and non-zero, then each agent rotates on a circle of radius $|\omega_0|^{-1}$. The direction of rotation around the circle is determined by the sign of $\omega_0$. If $\omega_0 > 0 $, then all the agents rotate in the anticlockwise direction and if $\omega_0 < 0$, then all the agents rotate in the clockwise direction.

Note that the agent's model, given by \eqref{modelNew}, is a unicycle model, and is widely studied in the literature \cite{Beard2008}, \cite{Mesbhai2010}, \cite{Zqu2009} in the context of modeling a real autonomous vehicle. We assume that the agents are identical and can exchange information about their orientations $\theta_k$ with all other agents of the group. Moreover, the control algorithms proposed in this paper are decentralized, and do not rely on any centralized information causing the agents to achieve phase balancing with a desired arrangement of their velocity directions. Only the heterogeneity in the controller gains is a mean to steer the agents towards phase balancing in a desired arrangement of their velocity directions.

\subsection{Notations}
We introduce a few notations, which are used in this paper. We use the bold face letters $\pmb{r} = [r_1, \ldots, r_N]^T \in \mathbb{C}^N$, $\pmb{\theta} = [\theta_1, \ldots, \theta_N]^T \in \mathbb{T}^N$, where $\mathbb{T}^N$ is the $N$-torus, which is equal to $\mathbb{S}^1 \times \ldots \times \mathbb{S}^1$ ($N$-times) to represent the vectors of length $N$ for the agent's positions and heading angles, respectively. Next, we define the inner product $\left<z_1, z_2\right>$ of the two complex numbers $z_1, z_2 \in \mathbb{C}$ as $\left<z_1, z_2\right> = \textrm{Re}(\bar{z}_1z_2)$, where, $\bar{z}_1$ represents the complex conjugate of $z_1$, and $\textrm{Re}(z)$ denotes the real part of $z \in \mathbb{C}$. This inner product is equivalent to the standard inner product on $\mathbb{R}^2$ since for some $z_k = x_k + iy_k, k=1,2$, the inner product $\left<z_1, z_2\right> = \text{Re}\{(x_1 - iy_1)(x_2 + iy_2)\} = x_1x_2 + y_1y_2$. For vectors, we use the analogous boldface notation $\left<\pmb{w}, \pmb{z}\right> = \textrm{Re}(\pmb{w^*}\pmb{z})$ for $\pmb{w}, \pmb{z} \in \mathbb{C}^N$, where $\pmb{w^*}$ denotes the conjugate transpose of $\pmb{w}$.


\subsection{Background and Problem Formulation}
At first, our prime requirement is to stabilize the motion of all the agents in a balanced formation. For this, we propose the feedback control $u_k, \forall k$, which is obtained by controlling the average linear momentum of the group of agents. The average linear momentum $p_\theta$ of the group of agents satisfying \eqref{modelNew1} is,
\begin{equation}
\label{phase_order_parameter}p_\theta = \frac{1}{N}\sum_{k=1}^{N}e^{i\theta_k} = \left|p_\theta\right|e^{i\Psi},
\end{equation}
which is also referred to as the phase order parameter \cite{Strogatz2000}. The modulus of the phase order parameter $\left|p_\theta\right|$ satisfies $0 \leq \left|p_\theta\right| \leq 1$, and is a measure of synchrony of the phase variable $\pmb{\theta}$. In particular, $\left|p_\theta\right| = 1$ for synchronized phases and $p_\theta = 0$ for balanced phases \cite{Strogatz2000}. Note that since $\theta_k, \forall k$, is a function of time $t$, $p_\theta$ varies with time, and we often suppress their time argument.

As mentioned before, in balanced formation, the position centroid (of the group of agents defined in \eqref{modelNew1}),
\begin{equation}
\label{R}R = \frac{1}{N}\sum_{k=1}^{N}r_k
\end{equation}
remains fixed, which implies that the quantity
\begin{equation}
\label{R_dot}\dot{R} = \frac{1}{N}\sum_{k=1}^{N}\dot{r}_k = \frac{1}{N}\sum_{k=1}^{N}e^{i\theta_k} = p_\theta
\end{equation}
is zero. Thus, the phase arrangement $\pmb{\theta}$ is balanced if the phase order parameter \eqref{phase_order_parameter} equals zero. This suggests that the stabilization of balanced formation is accomplished by considering the potential
\begin{equation}
\label{potential function}U(\pmb{\theta}) = \frac{N}{2}|p_\theta|^2,
\end{equation}
which is minimized when $p_\theta = 0$ (balanced formation). When $|p_\theta| = 1$, it corresponds to synchronized formation of the agents, and has been discussed in \cite{Jain2016}. Note that, with unit mass assumption, the position centroid $R$, given by \eqref{R}, is also the center of mass of the group of agents.

Now, we state the following theorem, which describe a Lyapunov-based control framework to stabilize balanced formation of the agents.

\begin{thm}\label{Theorem1}
Consider the system dynamics \eqref{modelNew} with control law
\begin{equation}
\label{control1}u_k = -K_k\left(\frac{\partial U}{\partial \theta_k}\right);~~~K_k \neq 0,
\end{equation}
and define a term
\begin{equation}
\label{term}T_k(\pmb{\theta}) = \left(\frac{\partial U}{\partial \theta_k}\right)^2
\end{equation}
for all $k = 1, \ldots, N$. If $\sum_{k=1}^{N} K_k T_k(\pmb{\theta}) > 0$, all the agents asymptotically stabilize to a balanced formation. Moreover, $K_k > 0, \forall k$, is a restricted sufficient condition in stabilizing balanced formation.
\end{thm}

\begin{proof}
Consider the potential function $U(\pmb{\theta})$ defined by \eqref{potential function}. Since the magnitude of the average linear momentum $|p_\theta|$ in \eqref{phase_order_parameter} satisfies $0 \leq |p_\theta| \leq 1$, it ensures that $0 \leq U(\pmb{\theta}) \leq N/2$. Also, the potential $U(\pmb{\theta})$ attains its minimum value only in the balanced formation, that is, $U(\pmb{\theta}) = 0$ only when $p_\theta = 0$. Thus, $U(\pmb{\theta})$ can be used as a Lyapunov function candidate \cite{Khalil2000}.

The time derivative of $U(\pmb{\theta})$, along the dynamics \eqref{modelNew}, is
\begin{equation}
\label{Udot}\dot{U}(\pmb{\theta}) = \sum_{k=1}^{N}\left(\frac{\partial U}{\partial \theta_k}\right) \dot{\theta}_k  = \sum_{k=1}^{N}\left(\frac{\partial U}{\partial \theta_k}\right) u_k.
\end{equation}
Using \eqref{control1} and \eqref{term}
\begin{equation}
\label{U_dot}\dot{U}(\pmb{\theta}) = -\sum_{k=1}^{N}K_k\left(\frac{\partial U}{\partial \theta_k}\right)^2 = -\sum_{k=1}^{N}K_k T_k(\pmb{\theta}),
\end{equation}
which shows that $\dot{U}(\pmb{\theta}) < 0$, if $\sum_{k=1}^{N}K_k T_k(\pmb{\theta}) > 0$. According to the Lyapunov stability theorem \cite{Khalil2000}, all the solutions of \eqref{modelNew} with the control \eqref{control1} asymptotically stabilize to the relative equilibrium where $U(\pmb{\theta})$ attains its minimum value, that is, at $p_\theta = 0$ (balanced formation).

The restricted sufficiency condition is proved next. Note that the term $T_k(\pmb{\theta}) \geq 0$ for all $k = 1, \ldots, N$, which ensures that $\dot{U}(\pmb{\theta}) \leq  0$ for $K_k > 0, \forall k$. Moreover, $\dot{U}(\pmb{\theta}) = 0$ if and only if $({\partial U}/{\partial \theta_k}) = 0, \forall k$, which defines the critical points of $U(\pmb{\theta})$. The critical set of $U(\pmb{\theta})$ is the set of all $\pmb{\theta} \in \mathbb{T}^N$, for which $({\partial U}/{\partial \theta_k}) = 0,~\forall k$. Since $\pmb{\theta} \in \mathbb{T}^N$ is compact, it follows from the LaSalle's invariance theorem \cite{Khalil2000} that all solutions of \eqref{modelNew}, under control \eqref{control1}, converge to the largest invariant set contained in $\{\dot{U}(\pmb{\theta}) = 0\}$, that is, the set
\begin{equation}
\Lambda = \left\{\pmb{\theta}~|~ ({\partial U}/{\partial \theta_k}) = \left<p_\theta, ie^{i\theta_k}\right> = 0,~\forall k\right\},
\end{equation}
which is the critical set of $U(\pmb{\theta})$. In this set, dynamics \eqref{modelNew2} reduces to $\dot{\theta}_k = 0, \forall k$, which implies that all the agents move in a straight line. The set $\Lambda$ is itself invariant since
\begin{align}
\nonumber \frac{d}{dt}\left<p_\theta, ie^{i\theta_k}\right> &= \left<p_\theta, \frac{d(ie^{i\theta_k})}{dt}\right> + \left<\frac{dp_\theta}{dt}, ie^{i\theta_k}\right>\\
&= -\left<p_\theta, e^{i\theta_k}\right>\dot{\theta}_k + \frac{1}{N}\left<\sum_{k=1}^{N} ie^{i\theta_k} \dot{\theta}_k, ie^{i\theta_k}\right> = 0
\end{align}
on this set. Therefore, all the trajectories of the system \eqref{modelNew} under control \eqref{control1} asymptotically converges to the critical set of  $U(\pmb{\theta})$. Moreover, the balanced state characterizes the stable equilibria of the system \eqref{modelNew} in the critical set $\Lambda$ and the rest of the critical points are unstable equilibria, which is proved next.

{\it Analysis of the critical set}: The critical points of $U(\pmb{\theta})$ are given by the $N$ algebraic equations
\begin{equation}
\label{critical}\frac{\partial U}{\partial \theta_k} = \left<p_\theta, ie^{i\theta_k}\right> =  |p_\theta|\sin(\Psi-\theta_k) = 0,~~1\leq k \leq N,
\end{equation}
where, $p_\theta = |p_\theta|e^{i\Psi}$, as defined in \eqref{phase_order_parameter}, has been used. Since the critical points with $p_\theta = 0$ are the global minima of $U(\pmb{\theta})$, the phase balancing is asymptotically stable if $K_k > 0, \forall k$.

Now, we focus on the critical points for which $p_\theta \neq 0$ and $\sin(\Psi - \theta_k) = 0, \forall k$. This implies that $\theta_k \in \{\Psi~\text{mod}~2\pi, (\Psi + \pi)~\text{mod}~2\pi\}, \forall k$. Let $\theta_k = (\Psi + \pi)~\text{mod}~2\pi$ for $k \in \{1,\ldots,M\}$, and $\theta_k = \Psi~\text{mod}~2\pi$ for $k \in \{M+1,\ldots, N\}$. Note that the value $M=0$ defines synchronized state ($|p_\theta| = 1$) and corresponds to the global maximum of $U(\pmb{\theta})$, and hence unstable if $K_k >0, \forall k$. Every other value of $1 \leq M \leq N-1$ such that $p_\theta \neq 0$ corresponds to the saddle point, and is, therefore, unstable for $K_k >0, \forall k$. This is proved below.

Let $H(\pmb{\theta}) = [h_{jk}(\pmb{\theta})]$ be the Hessian of $U(\pmb{\theta})$. Then, we can find the components $[h_{jk}(\pmb{\theta})]$ of $H(\pmb{\theta})$ by evaluating the second derivatives $\frac{\partial^2 U}{\partial\theta_j \partial\theta_k}$ for all pairs of $j$ and $k$, which yields
\[
    h_{jk}(\pmb{\theta})=
\begin{cases}
    \dfrac{1}{N} - \left<p_\theta, e^{i\theta_k}\right> = \dfrac{1}{N} -|p_\theta|\cos(\Psi-\theta_k), &  j = k\\
    \dfrac{1}{N}\left<e^{i\theta_j}, e^{i\theta_k}\right> = \dfrac{1}{N} \cos(\theta_j-\theta_k),      &  j \neq k.
\end{cases}
\]
Since $\theta_k = (\Psi + \pi)~\text{mod}~2\pi$ for $k \in \{1,\ldots,M\}$, and $\theta_k = \Psi~\text{mod}~2\pi$ for $k \in \{M+1,\ldots, N\}$, $\cos(\Psi-\theta_k) = 1$ for $k \in \{1,\ldots,M\}$, and $\cos(\Psi-\theta_k) = -1$ for $k \in \{M+1,\ldots, N\}$. Hence, the diagonal entries ($j = k$) of the Hessian $H(\pmb{\theta})$ are given by
\[
    h_{kk}(\pmb{\theta})=
\begin{cases}
    (1/N) + |p_\theta|, &  k\in\{1, \ldots, M\}\\
    (1/N) - |p_\theta|, & k\in\{M+1, \ldots, N\},
\end{cases}
\]
where, $1 \leq M \leq N-1$. Since $(1/N) + |p_\theta| > 0$, the Hessian matrix $H(\pmb{\theta})$ has at least one positive pivot, and hence one positive eigenvalue \cite{strang2007}. In order to show that all critical points $1 \leq M \leq N-1$ such that $p_\theta \neq 0$, are saddle points, we verify that the Hessian matrix $H(\pmb{\theta})$ is indefinite by showing that it has at least one negative eigenvalue.

Since $\theta_k$ is as given above, $\cos(\theta_j-\theta_k) = 1$ for $j,k\in\{1, \ldots, M\}~\text{or}~j,k\in\{M+1, \ldots, N\}$, and $\cos(\theta_j-\theta_k) = -1$ for $j \in\{1, \ldots, M\}, k\in\{M+1, \ldots, N\}~\text{or}~j \in \{M+1, \ldots, N\}, k \in \{1, \ldots, M\}$. Hence, the off diagonal entries ($j \neq k$) of $H(\pmb{\theta})$ are given by
\[
    h_{jk}(\pmb{\theta})=
\begin{cases}
    (1/N), &  \left.
    \begin{array}{l}
      j,k\in\{1, \ldots, M\}\\
      \text{or}~j,k\in\{M+1, \ldots, N\}
    \end{array}
  \right.\\
    -(1/N), & \text{otherwise}.
\end{cases}
\]
Define a vector $\pmb{w} = [w_1, \ldots, w_M, -w_{M+1}, \ldots, -w_N]^T$, with $w_k = 1, \forall k$. Then, the Hessian $H(\pmb{\theta})$ can be written in a compact form as
\begin{equation}
H(\pmb{\theta}) = \frac{1}{N}\pmb{w}\pmb{w}^T + |p_\theta|\text{diag}(\pmb{w}),
\end{equation}
where, $\text{diag}(\pmb{w})$ is a diagonal matrix whose diagonal entries are given by the entries of the vector $\pmb{w}$. Now, define a vector $\pmb{q} = [q_1, \ldots, q_N]^T$ with $q_k = 0, k = {1, \ldots, N-2}$, and $q_{N-1} = -1$ and $q_{N} = 1$. By construction, $\pmb{w}^T\pmb{q} = 0$ and hence,
\begin{equation}
\pmb{q}^TH(\pmb{\theta})\pmb{q} = |p_\theta|\pmb{q}^T\text{diag}(\pmb{w})\pmb{q} = -2|p_\theta| < 0,
\end{equation}
which shows that $H(\pmb{\theta})$ is an indefinite matrix. Hence, the critical points satisfying $\sin(\Psi-\theta_k) = 0, \forall k$, along with $p_\theta \neq 0$ are saddle points and are unstable when $K_k > 0, \forall k$. This completes the proof.
\end{proof}

It is evident in the Theorem~\ref{Theorem1} that the condition $\sum_{k=1}^{N}K_k T_k(\pmb{\theta}) > 0$ is yet satisfied if at least one of the heterogeneous gains is non-zero and positive, and all other gains are zero. However, if the control gains are zero for more than a certain number of agents, no balanced formation may be achieved under the control law \eqref{control1}. For instance consider the case of $N=3$, suppose the control gains are zero for the agents 1 and 2, and is positive for the agent 3. In this situation, since the control force \eqref{control1} is zero for the agents 1 and 2, these agents will keep on moving in the directions of initial heading angles $\theta_1(0)$, and $\theta_2(0)$, respectively, and hence, phase balancing of agents may not be achievable. Nonetheless, if the heterogeneous gains are zero at most for $\lfloor N/2 \rfloor$ agents, where, $\lfloor N/2 \rfloor$ is the largest integer less than or equal to $N/2$, balanced formation can be achieved under the control law \eqref{control1}. This is proved in the following corollary.

\begin{cor}\label{cor0}
For the conditions given in Theorem~\ref{Theorem1}, if the heterogenous control gains $K_k,~k = 1, \ldots, N$, are zero at most for $\lfloor N/2 \rfloor$ agents and positive for rest of the agents, balanced formation is asymptotically stable under the control law \eqref{control1}.
\end{cor}

\begin{proof}
From \eqref{critical}, since the critical points where $p_\theta \neq 0$ and $\sin(\Psi - \theta_k) = 0, \forall k$, are characterized by $M$ synchronized phases at $(\Psi + \pi)~\text{mod}~2\pi$, and $N-M$ synchronized phases at $\Psi~\text{mod}~2\pi$, with $1 \leq M \leq N-1$, the phases $\theta_k$ necessarily lie in one of two clusters that are on opposite sides of the unit circle. Moreover, all of the phases within each cluster are identical. The clue of the proof lies in the fact that, unlike analyzing critical point for each $1 \leq M \leq N-1$, it is sufficient to check only for $1 \leq  M  <  N/2$ since, for $N/2 <  M \leq N-1$, the clusters of the phases $\theta_k$, on the unit circle, replicates. Since $N/2 \leq N-1$, it follows from Theorem~\ref{Theorem1} that every $1 \leq M  <  N/2$ corresponds to a saddle point. Therefore, if the heterogeneous gains are zero for the rest $\lfloor N/2 \rfloor$ agents, balanced formation is asymptotically stable. This completes the proof.
\end{proof}

Next, we state a corollary, which ensures the stabilization of agents in the balanced formation when they move at an angular velocity $\omega_0$ around individual circular orbits.

\begin{cor}\label{cor1}
Theorem~\ref{Theorem1} holds for the system dynamics \eqref{modelNew}, under the control law, given by
\begin{equation}
\label{control2}u_k = \omega_0 - K_k\left(\frac{\partial U}{\partial \theta_k}\right)
\end{equation}
for all $k=1, \ldots, N$.
\end{cor}

\begin{proof}
Under the control \eqref{control2}, the time derivative of $U(\pmb{\theta})$ along the dynamics \eqref{modelNew} is
\begin{equation}
\label{U_dot_New}\dot{U}(\pmb{\theta}) = \omega_0\sum_{k=1}^{N}\frac{\partial U}{\partial \theta_k} - \sum_{k=1}^{N}K_k\left(\frac{\partial U}{\partial \theta_k}\right)^2
\end{equation}

Note that
\begin{equation}
\label{relation}\sum_{k=1}^{N} \frac{\partial U}{\partial \theta_k} = \sum_{k=1}^{N}\left<p_\theta, ie^{i\theta_k}\right> = \frac{1}{N}
\sum_{k=1}^{N}\sum_{j=1}^{N} \sin(\theta_j - \theta_k) = 0
\end{equation}

Using \eqref{relation}, \eqref{U_dot_New} can be rewritten as
\begin{equation}
\label{Udot_final}\dot{U}(\pmb{\theta}) = -\sum_{k=1}^{N}K_k T_k(\pmb{\theta}),
\end{equation}
which is the same as \eqref{Udot}. Therefore, the conclusions of Theorem~\ref{Theorem1} are unchanged under control \eqref{control2}.
\end{proof}


\begin{cor}\label{cor2}
Under the conditions given in Theorem~\ref{Theorem1}, the magnitude $|p_{\theta}|$ of the phase order parameter $p_{\theta}$, given by \eqref{phase_order_parameter}, is strictly decreasing with time if $\sum_{k=1}^{N}K_k T_k(\pmb{\theta}) > 0$, and non-increasing with time if $K_k > 0, \forall k$.
\end{cor}

\begin{proof}
The time derivative of \eqref{potential function} yields
\begin{equation}
\dot{U}(\pmb{\theta}) = N|p_\theta|\frac{d|p_\theta|}{dt},
\end{equation}
which by using \eqref{U_dot} can be rewritten as
\begin{equation}
\label{decresing}|p_\theta|\frac{d|p_\theta|}{dt} = -\frac{1}{N}\sum_{k=1}^{N} K_kT_k(\pmb{\theta}).
\end{equation}
According to Theorem~\ref{Theorem1}, since the agents asymptotically stabilize to a balanced formation, that is, $p_\theta \rightarrow 0$ as $t \rightarrow \infty$, $p_\theta \neq 0$ for all intermediate times, and hence, \eqref{decresing} can be rewritten as
\begin{equation}
\frac{d|p_\theta|}{dt} = -\frac{1}{N|p_\theta|}\sum_{k=1}^{N} K_kT_k(\pmb{\theta}),
\end{equation}
which implies that ${d|p_\theta|}/{dt} < 0$ (i.e, $|p_\theta|$ is strictly decreasing) whenever $\sum_{k=1}^{N}K_k T_k(\pmb{\theta}) > 0$, and ${d|p_\theta|}/{dt} \leq 0$ (i.e, $|p_\theta|$ is non-increasing) if $K_k > 0, \forall k$. This completes the proof.
\end{proof}


\subsection{Problem Description}

Now, we formally state the main objective of this paper. The control law, given by \eqref{control2}, can be written as
\begin{equation}
\label{control3}\dot{\theta_k} = \omega_0 - \frac{K_k}{N} \sum_{j=1}^{N} \sin(\theta_j - \theta_k).
\end{equation}
The term $K_k$ is the control gain of the $k^\text{th}$ agent. Prior work in \cite{Sepulchre2007} uses the same control gain $K_k = K,\forall k$, whereas we extend the analysis by using different gains $K_k$ for different agents. This is the heterogeneous control gains case of interest to us in this paper. Unlike the case of synchronization in \cite {Jain2016}, the analysis for $N$ agents is quite involved in the case of balanced formation. Therefore, in this work, the analytical results for heterogeneous gains are given mainly for two and three-agent systems. However, the results for $N$ agents, have been presented though simulations.

\begin{remark}
Note that, in the Theorem~\ref{Theorem1}, the conditions $\sum_{k=1}^{N}K_k T_k(\pmb{\theta}) > 0 $ may be satisfied for both positive and negative values of gains $K_k$ because of the involvement of the term $T_k(\pmb{\theta})$. However, in this paper, the idea of introducing heterogeneous gains is illustrated mainly for the restrictive sufficient condition on $K_k$, that is, $K_k > 0, \forall k$, since the analysis for the set of gains $K_k$ satisfying $\sum_{k=1}^{N}K_k T_k(\pmb{\theta}) > 0$ is quite involved for $N > 2$. Moreover, it will be shown for $N=2$ that the reachable set of the velocity directions of the agents in balanced formation further expands for the controller gains $K_k$ satisfying the condition $\sum_{k=1}^{N}K_k T_k(\pmb{\theta}) > 0 $.
\end{remark}

\begin{figure*}
\centering
\subfigure[]{\includegraphics[scale=0.65]{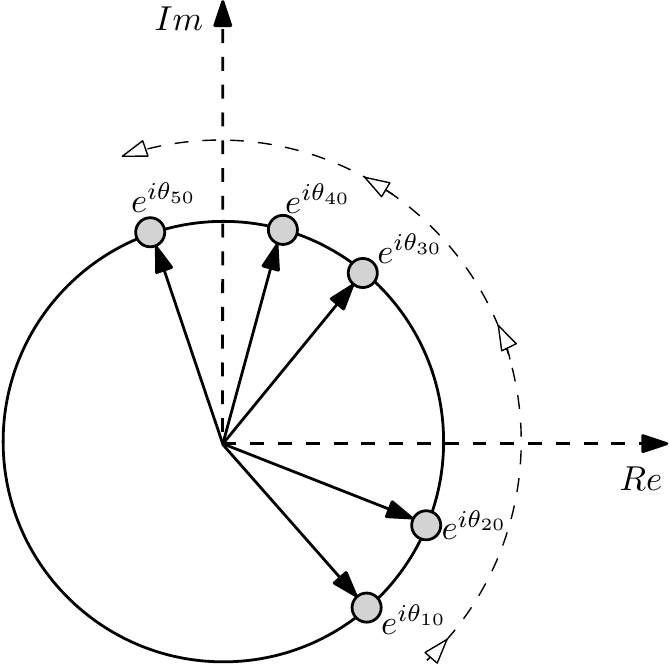}}\hspace{1cm}
\subfigure[]{\includegraphics[scale=0.65]{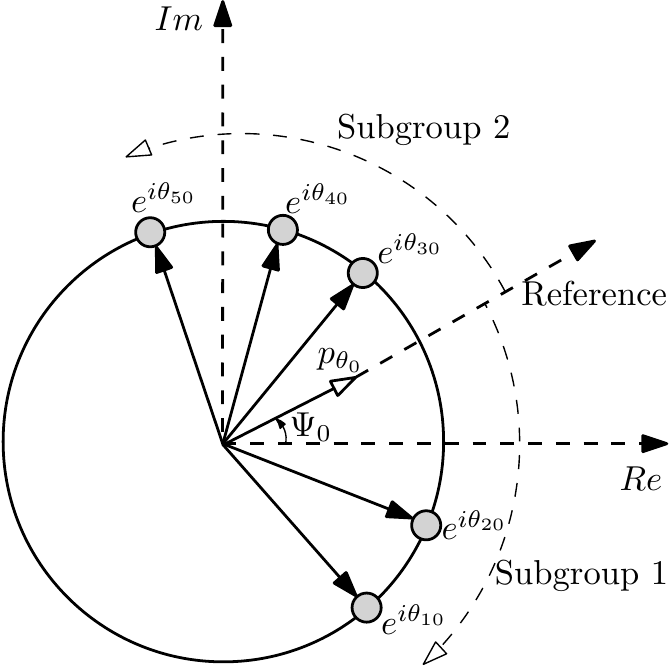}}\hspace{1cm}
\subfigure[]{\includegraphics[scale=0.65]{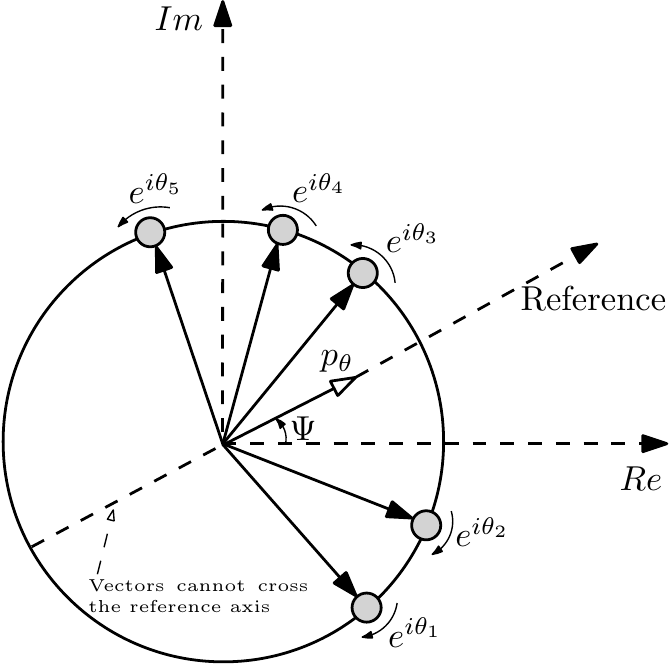}}
\caption{Representation of the velocity vectors of the agents around the unit circle for $N=5$. $(a)$ Cyclic arrangement of the initial unit velocity vectors $e^{i\theta_{k0}}, k = 1, \ldots, 5$. $(b)$ Distribution of vectors in two subgroups with respect to the initial phase order parameter vector $p_{\theta_0}$. $(c)$  Dynamics of the unit vectors at a particular instant in time $t$ against heterogeneous controller gains $K_k, k = 1, \ldots, 5$, given according to the Assumption~\ref{assumption}}.
\label{Dynamics of agents}
\end{figure*}

\section{Reachable Velocity Directions}
The velocity directions that are achievable in balanced formation of agents for different values of the heterogeneous controller gains are called the reachable velocity directions.

In this section, we analyze the reachability condition of the velocity directions of agents in balanced formation against heterogeneous controller gains $K_k > 0, \forall k$. In particular, it will be shown for $N \in \{2, 3\}$ that a desired arrangement of velocity vectors of the agents in balanced formation can be obtained by suitably choosing the heterogeneous control gains.

At first, we discuss the results for $\omega_0 = 0$. Then, we extend these results to $\omega_0 \neq 0$ by performing the analysis in a rotating frame of reference.

\subsection{Case~1: $\omega_0 = 0$}
For $\omega_0 = 0$, the control law, given by \eqref{control3}, can be written as
 \begin{equation}
\label{control3_1}\dot{\theta_k} = - \frac{K_k}{N} \sum_{j=1}^{N} \sin(\theta_j - \theta_k).
\end{equation}

Let the agents, with dynamics given by \eqref{modelNew}, start from initial heading angles $\pmb{\theta}(0) = [\theta_{10}, \ldots, \theta_{N0}]^T \in (-\pi, \pi)^N$, where, $\theta_{10} < \theta_{20} < \ldots < \theta_{N0}$. This ensures that the initial velocity vectors $e^{i\theta_{10}}, \ldots, e^{i\theta_{N0}}$ of the agents, can be arranged on the unit circle (in a complex plane) in a cyclic manner in which the successor of agent $k$ is agent $k+1$ modulo $N$ (counted in the anticlockwise direction). This is illustrated in Fig.~$1(a)$ for the five agents.

For the initial heading angles $\pmb{\theta}(0)$ of the agents as defined above, the initial phase order vector $p_{\theta}(\pmb{\theta}(0))$, from \eqref{phase_order_parameter}, is given by
\begin{equation}
\label{initial_phase_order}p_{\theta}(\pmb{\theta}(0)) = \frac{1}{N}\sum_{k=1}^{N} e^{i\theta_{k0}} \triangleq {p}_{\theta_0} = |{p}_{\theta_0}|e^{i\Psi_0},
\end{equation}
which is the resultant of all equally scaled initial velocity vectors $(1/N)e^{i\theta_{10}}, \ldots, (1/N)e^{i\theta_{N0}}$. Since $0 \leq |{p}_{\theta_0}| \leq 1$, the vector ${p}_{\theta_0} \in S_z$, where, $S_z = \{z \in \mathbb{C}~{\big |}~|z| \leq 1\}$ is the set of all the points residing in the interior and on the boundary of a unit circle in the complex plane, and is shown in Fig.~$1(b)$ for the arrangement of unit vectors in Fig.~$1(a)$.

In order to obtain the results and for the ease of analysis in this paper, we divide the initial velocity vectors $e^{i\theta_{10}}, \ldots, e^{i\theta_{N0}}$ of $N$ agents in two subgroups with respect to the the initial phase order vector ${p}_{\theta_0}$. This is done by choosing a reference axis along the initial phase order parameter ${p}_{\theta_0}$, as shown in Fig.~$1(b)$, and measuring the angle of each unit vector $e^{i\theta_{k0}}, \forall k$, with respect to it. By doing so, let there be $\overline{N} < N$ agents in subgroup~1 for which $0 < \Psi_0 - \theta_{k0} < \pi$ (that is, for unit vectors lying in the clockwise direction of ${p}_{\theta_0}$), where, $k \in \{1,\ldots, \overline{N}\}$, and $N-\overline{N}$ agents in subgroup~2 for which $-\pi < \Psi_0 - \theta_{k0} < 0$ (that is, for unit vectors lying in the anticlockwise direction of ${p}_{\theta_0}$), where, $k \in \{\overline{N}+1,\ldots,N\}$.

Based on these notations, the following mild assumption on the heterogeneous control gains, in addition to all being positive, has been taken into account in order to solve the problem addressed in this paper.

\begin{assumption}\label{assumption}
Corresponding to each agent in a particular subgroup as described above, the heterogeneous controller gains are chosen in a fashion such that they are non-decreasing as we move from the initial phase order parameter vector $p_{\theta_0}$ towards the furthermost initial velocity vector $e^{i\theta_{k0}}$ in individual subgroups. For example, the heterogeneous control gains for the scenario as shown in Fig.~$1(b)$, are chosen such that they satisfy $K_1 \geq K_2 \geq \ldots \geq K_{\overline{N}} \geq 0$ for the agents in subgroup~1 and $0 \leq K_{\overline{N}+1} \leq K_{\overline{N}+2} \leq \ldots \leq K_N$ for the agents in subgroup~2.
\end{assumption}

\begin{remark}
Suppose if $\Psi_0 - \theta_{k0} = 0$ for a particular $k^\text{th}$ agent then, the corresponding controller gain $K_k$ is independently chosen, and can assume any positive value.
\end{remark}

Based on these assumptions and notations, we now state the following lemma which depicts the behavior of final velocity vectors of the agents in balanced formation against heterogeneous control gains.

\begin{lem}\label{lem1}
Consider $N$ agents, with dynamics given by \eqref{modelNew}, under the control law \eqref{control3_1} with heterogeneous control gains $K_k, \forall k$, given according to the Assumption~\ref{assumption}. Let the initial heading angles of the agents be given by $\pmb{\theta}(0)$ such that the initial velocity vectors $e^{i\theta_{10}}, \ldots, e^{i\theta_{N0}}$ are in cyclic order with  $\theta_{10} < \theta_{20} < \ldots < \theta_{N0}$. Then, the cyclic arrangement of agents' velocity vectors, $e^{i\theta_{1f}}, \ldots, e^{i\theta_{Nf}}$, in balanced formation, is retained on the unit circle, where, $\theta_{kf} = \theta_k(t\rightarrow\infty)$, denotes the orientation of the $k^\text{th}$ agent in the steady-state.
\end{lem}

\begin{proof}
Without loss of generality, the proof of this lemma is provided with reference to Fig.~\ref{Dynamics of agents}.

From \eqref{phase_order_parameter}, we can write
\begin{equation}
\left|p_{\theta}\right|e^{i(\Psi - \theta_k)} =  \frac{1}{N}\sum_{j=1}^{N}e^{i(\theta_j - \theta_k)},
\end{equation}
the imaginary part of which is given by
\begin{equation}
\label{phase_order_parameter_New}\left|p_{\theta}\right|\sin(\Psi - \theta_k) = \frac{1}{N}\sum_{j=1}^{N} \sin(\theta_j - \theta_k)
\end{equation}
Using \eqref{phase_order_parameter_New}, \eqref{control3_1} can be written as
\begin{equation}
\label{theta_dot_new}\dot{\theta}_k = -K_k\left|p_{\theta}\right|\sin(\Psi - \theta_k),
\end{equation}
which implies that the heading angle $\theta_k$ of the $k^\text{th}$ agent moves away from the average phase $\Psi$ of the whole ensemble. The interpretation of the dynamics \eqref{theta_dot_new}, at an instant in time $t$, is shown in Fig.~$1(c)$.

For better understanding of the dynamics \eqref{theta_dot_new}, $\forall k$, and $\forall t$, it is convenient to choose the reference axis along the phase order parameter $p_{\theta}$, as shown in Fig.~$1(c)$, and measure the angle of each unit vector with respect to it. By doing so, it is easy to see that $\left|\Psi - \theta_k\right| < \pi$ for all $k = 1, \ldots, N$. Therefore, for $K_k > 0, \forall k$, one can observe from \eqref{theta_dot_new} that, if $0 < \Psi - \theta_k < \pi$ (that is, for unit vectors lying in the clockwise direction of $p_{\theta}$), $\dot{\theta}_k < 0$, and if $-\pi < \Psi - \theta_k < 0$ (that is, for unit vectors lying in the anticlockwise direction of $p_{\theta}$), $\dot{\theta}_k > 0$. It means that the heading angle of the $k^\text{th}$ agent always moves away from the average phase $\Psi$ of the group, and hence the angular separation $\left|\Psi - \theta_k\right|, \forall k$, increases with time. Moreover, the maximum value of $\left|\Psi - \theta_k\right|, \forall k$, is $\pi$ radians since whenever $|\Psi - \theta_k| > \pi$, the sign of $\sin(\Psi - \theta_k)$ in \eqref{theta_dot_new} changes, and hence the unit vector $e^{i\theta_k}$ of the $k^\text{th}$ agent now starts moving in the opposite direction, and hence cannot cross the reference axis.

Let the agents be divided in two subgroups as discussed above, and the corresponding heterogeneous gains $K_k, \forall k$, are given according to the Assumption~\ref{assumption}. This is clear from Fig.~$1(b)$ that, at time instant $t=0$, it holds that $\left|\Psi_0 - \theta_{10}\right| > \left|\Psi_0 - \theta_{20}\right| > \ldots > \left|\Psi_0 - \theta_{\overline{N}0}\right|$ for the agents in subgroup~1 and $\left|\Psi_0 - \theta_{(\overline{N}+1)0}\right| < \left|\Psi_0 - \theta_{(\overline{N}+2)0}\right| < \ldots < \left|\Psi_0 - \theta_{N0}\right|$ for the agents in subgroup~2. Therefore, under the influence of control gains $K_k, \forall k$, given according to the Assumption~\ref{assumption}, it now follows from \eqref{theta_dot_new} that $|\dot{\theta}_1(t)| \geq |\dot{\theta}_2(t)| \geq \ldots \geq |\dot{\theta}_{\overline{N}}(t)|$ for the agents in subgroup~1, and $|\dot{\theta}_{\overline{N}+1}(t)| \leq |\dot{\theta}_{\overline{N}+2}(t)| \leq \ldots \leq |\dot{\theta}_N(t)|$ for the agents in subgroup~2 for all $t$. This turns out that $\left|\Psi(t) - \theta_{1}(t)\right| > \left|\Psi(t) - \theta_{2}(t)\right| > \ldots > \left|\Psi(t) - \theta_{\overline{N}}(t)\right|$ for the agents in subgroup~1 and $\left|\Psi(t) - \theta_{\overline{N}+1}(t)\right| < \left|\Psi(t) - \theta_{\overline{N}+2}(t)\right| < \ldots < \left|\Psi(t) - \theta_{N}(t)\right|$ for the agents in subgroup~2 for all $t$. As a result, the cyclic order of the agents' velocity vectors $e^{i\theta_{1f}}, \ldots, e^{i\theta_{Nf}}$, in the balanced formation, is maintained on the unit circle. This completes the proof.
\end{proof}

\begin{remark}
Note that the Lemma~\ref{lem1} comments only on the cyclic arrangement of the agents' velocity vectors $e^{i\theta_{1f}}, \ldots, e^{i\theta_{Nf}}$, in balanced formation, though, it doesn't tell anything about the numerical values of the orientations $\theta_{1f}, \ldots, \theta_{Nf}$, in a given coordinate frame.
\end{remark}

\begin{remark}
It is straight forward to see that whenever $\Psi_0 - \theta_{k0} = 0$ for a particular $k^\text{th}$ agent then, the corresponding controller gain $K_k$ can be independently chosen, and can assume any positive value since it coincides with the reference axis and hence, this does not affect the analysis of Lemma~\ref{lem1}.
\end{remark}

\begin{figure*}
\centering
\subfigure[]{\includegraphics[scale=0.65]{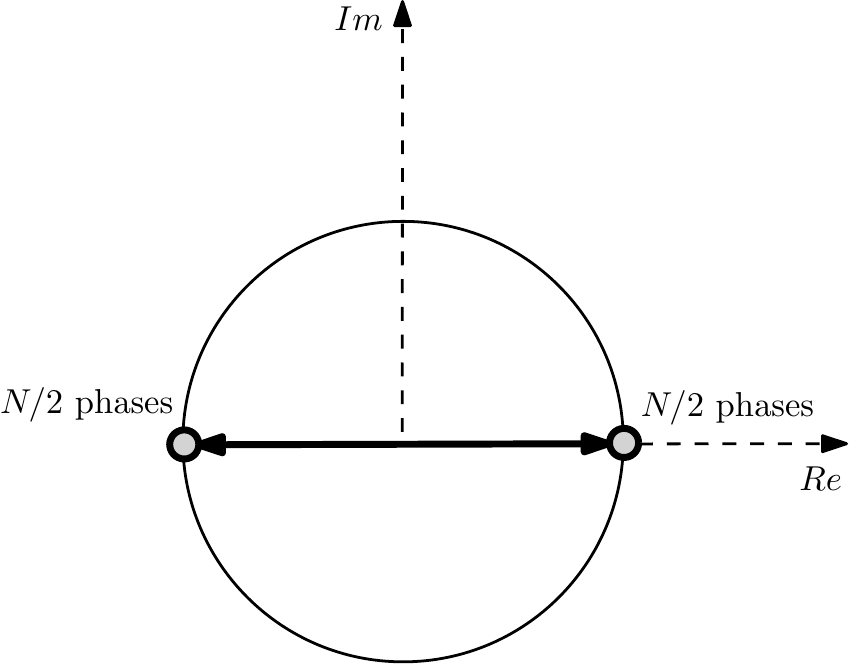}}\hspace{1cm}
\subfigure[]{\includegraphics[scale=0.65]{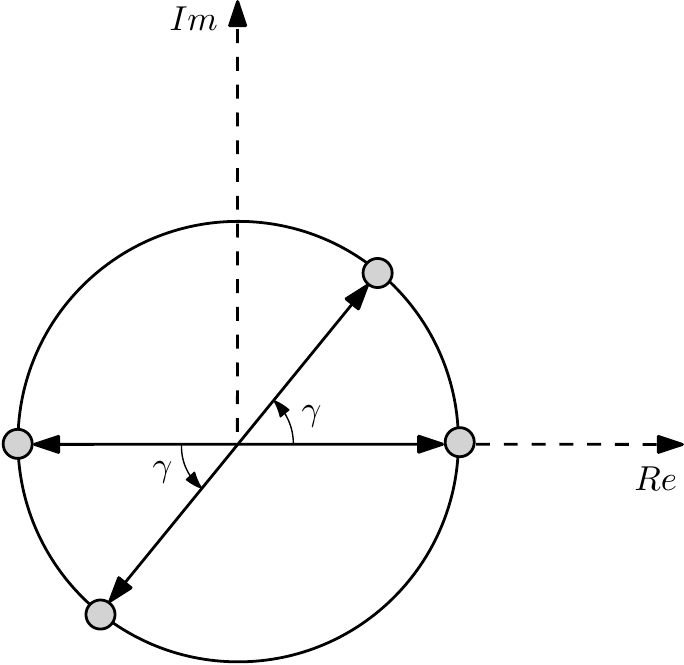}}
\caption{Various possible balanced formation for $N=4$. $(a)$ One cluster of $N/2$ agents at 0 and other at $\pi$ radians. $(b)$ Four clusters each with $N/4$ phases along with a mirror symmetry about both axes. Here, the angle $\gamma$ between the velocity vectors of two consecutive agents is not unique.}
\label{balanced formation for four agents}
\end{figure*}

Next, we describe the phase balancing of two and three agents, and prove that the angular separation between the velocity vectors of any two consecutive agents in balanced formation is unique. However, for $N > 3$, the equation $p_\theta = 0$, may be satisfied for the phase arrangements $\pmb{\theta} \in \mathbb{T}^N$ such that the angular separation between the velocity vectors of any two consecutive agents in the steady state may not be unique. For example, for $N=4$, $p_\theta = 0$ if there is $(i)$ a cluster of $N/2$ phases at $0$ radian, and another cluster of $N/2$ phases at $\pi$ radians or $(ii)$ four clusters-each with $N/4$ phases along with a mirror symmetry about both axes. This situation is shown in Fig.~\ref{balanced formation for four agents}, where, the angular separation $\gamma$ can assume any value provided clusters of $N/2$ phases are individually balanced. For instance, the value $\gamma = 0$ corresponds to the case $(i)$, and $\gamma = \pi/2$ corresponds to the splay phase arrangement. Nonetheless, in splay formation, the angular separation $\gamma = 2\pi/N$, is unique for all consecutive agents' pair, unlike Lemma~\ref{lem1}, it is yet challenging to set-up the cyclic arrangement of agents' velocity vectors in the steady state (see Appendix).

\begin{lem}\label{lem2}
Consider $N \in \{2, 3\}$ agents, with dynamics given by \eqref{modelNew}, under the control law \eqref{control3_1} with heterogeneous control gains $K_k, \forall k$, given according to the Assumption~\ref{assumption}. Let the initial heading angles of the agents be given by $\pmb{\theta}(0)$ such that the initial velocity vectors $e^{i\theta_{10}}, \ldots, e^{i\theta_{N0}}$, are in cyclic order with  $\theta_{10} < \theta_{20} < \ldots < \theta_{N0}$. Then, the phase balancing with two and three agents occurs if and only if their velocity directions in the steady-state are at an equal angular separation of $\pi$, and $2\pi/3$ radians, respectively.
\end{lem}

\begin{proof}
Let the velocity directions of agents in balanced formation be given by $\pmb{\theta}_f = [\theta_{1f}, \ldots, \theta_{Nf}]^T$, where, $\theta_{kf} = \theta_k (t \rightarrow \infty)$, denotes the orientation of the $k^\text{th}$ agent in the steady-state.

First consider the case of $N=2$. Since the rate of change of the position of the centroid is zero in balanced condition, we have $\dot{R} = 0$, and hence, by using \eqref{R_dot}, we can write
\begin{eqnarray}
\label{rel_1}\cos\theta_{1f} + \cos\theta_{2f} = 0~~~\text{and}~~~\sin\theta_{1f} + \sin\theta_{2f} = 0.
\end{eqnarray}
On squaring and adding \eqref{rel_1}, we get
\begin{eqnarray}
\label{rel_2}\cos(\theta_{2f} - \theta_{1f}) = -1.
\end{eqnarray}
The solution of \eqref{rel_2} is given by
\begin{equation}
\label{rel_3}\theta_{2f}-\theta_{1f} = (2n + 1)\pi,
\end{equation}
where, $n\in\mathbb{Z}$, ($\mathbb{Z}$ is a set of integers). Thus,
\begin{equation}
(\theta_{2f}-\theta_{1f})~(\text{mod}~2\pi) = \pi,
\end{equation}
that is, the agents are at an angular separation of $\pi$ radians. This proves the necessary condition. To prove the sufficiency condition, let us substitute $\theta_{1f} = \theta_f$, and $\theta_{2f} = \theta_f + \pi$ in \eqref{R_dot}, which results in $\dot{R} = p_\theta = 0$, and hence the agents are in balanced formation.

Next, we consider the case of $N=3$. Let the velocity directions of agents in phase balancing be related as
\begin{equation}
\label{count1} \theta_{1f} = \theta_f,~~ \theta_{2f} = \theta_f + \psi_1,~~ \theta_{3f} = \theta_f + \psi_2,
\end{equation}
where, $\psi_1$ and $\psi_2$ are the angular separations between the velocity directions of agents $1$ and $2$, and agents $1$ and $3$, respectively. Since for the given conditions on the controller gains and the initial heading angles, it follows from the Lemma~\ref{lem1} that the velocity vectors of agents are in cyclic order on the unit circle, it holds that $\psi_2 > \psi_1 > 0$.

Substituting \eqref{count1} in \eqref{R_dot}, we get
\begin{subequations}\label{count2}
\begin{align}
\cos \left(\theta_f+\psi_1\right) + \cos \left(\theta_f+\psi_2\right) =& -\cos\theta_f \label{eq4}\\
\sin \left(\theta_f+\psi_1\right) + \sin \left(\theta_f+\psi_2\right) =& -\sin\theta_f \label{eq5}
\end{align}
\end{subequations}
Squaring and adding \eqref{count2}, we get
\begin{equation}
\cos\left(\psi_2-\psi_1\right) = -{1}/{2},
\end{equation}
which implies that
\begin{equation}
\psi_2-\psi_1 = 2n\pi \pm \left({2\pi}/{3}\right),
\end{equation}
where, $n \in \mathbb{Z}$. Since the control force applied to an agent is zero when they are in phase balancing (Theorem~\ref{Theorem1}), the following expressions can be obtained by substituting \eqref{count1} in \eqref{control3_1}.
\begin{subequations}\label{count3}
\begin{align}
\label{eq1}\sin\psi_1 + \sin\psi_2 &= 0\\
\label{eq2}-\sin\psi_1 + \sin\left(\psi_2-\psi_1\right) &= 0\\
\label{eq3}-\sin\psi_2 + \sin\left(\psi_1-\psi_2\right) &= 0
\end{align}
\end{subequations}
Since the equation \eqref{eq1} is achievable by adding \eqref{eq2} and \eqref{eq3}, it is sufficient to solve \eqref{eq2} and \eqref{eq3} to get a solution for $\psi_1$ and $\psi_2$. Note that, in the phase balancing condition, \eqref{count2} and \eqref{count3} should satisfy simultaneously for some $\psi_1$ and $\psi_2$ such that $\psi_2 > \psi_1 >0$ as per our consideration. Thus, by substituting $\psi_2-\psi_1 = 2n\pi \pm {2\pi}/{3}$ in \eqref{count3}, and solve for $\psi_1$ and $\psi_2$ by properly choosing integer $n$ so that $\psi_2 > \psi_1 >0$, we get $\psi_1 = 2\pi/3$, and $\psi_2 = 4\pi/3$. This proves the necessary condition. The sufficiency condition can be proved similarly to the two agents' case, and hence this proof is omitted. This completes the proof.
\end{proof}

Following Lemma~\ref{lem2}, we may assume for $N \in \{2, 3\}$ that $\theta_{kf} = \theta_f + ({2(k-1)\pi}/{N}), \forall k$, where, $\theta_{kf} = \theta_k (t\rightarrow\infty)$, denotes the orientation of the $k^\text{th}$ agent in balanced formation, and $\theta_f \in \mathbb{S}^1$, represents a reference direction (which is the orientation of agent\#$1$ in the steady-state). Once the reference direction $\theta_f$ is known, we can easily determine the orientations of agents in balanced formation.

Based on these notations, we now state the following theorem, which says that, by using heterogeneous gains, it is possible to get a desired reference direction $\theta_f$, and consequently, the velocity directions of agents in balanced formation can be determined.

\begin{thm}\label{Theorem2}
Consider $N \in \{2, 3\}$ agents, with dynamics given by \eqref{modelNew}, under the control law \eqref{control3_1} with non-zero heterogeneous control gains $K_k, \forall k$, given according to the Assumption~\ref{assumption}. Let the initial heading angles of the agents be given by $\pmb{\theta}(0)$ such that the initial velocity vectors $e^{i\theta_{10}}, \ldots, e^{i\theta_{N0}}$, are in cyclic order with  $\theta_{10} < \theta_{20} < \ldots < \theta_{N0}$. Then, the velocity directions of the agents in balanced formation, are given by $\theta_{kf} = \theta_f + ({2(k-1)\pi}/{N}), \forall k$, where, the reference direction $\theta_f$, is
\begin{equation}
\label{theta_f}\theta_f =  \left\{\sum_{k=1}^{N}\frac{1}{K_k}\left(\theta_{k0} - \frac{2(k-1)\pi}{N}\right)\right\}{\Big /}\left\{\sum_{k=1}^{N}\frac{1}{K_k}\right\}.
\end{equation}
\end{thm}

\begin{proof}
Taking the summation on both the sides of \eqref{control3_1} over all $k=1, \ldots, N$, we get
\begin{equation}
\label{angle_relation}\sum_{k=1}^{N} \frac{\dot{\theta}_k(t)}{K_k} = -\frac{1}{N}\sum_{k=1}^{N}\sum_{j=1}^{N} \sin(\theta_j - \theta_k) = 0.
\end{equation}
Integration of \eqref{angle_relation} yields
\begin{equation}
\label{angle_relation1}\sum_{k=1}^{N} \frac{{\theta}_k(t)}{K_k} = \sum_{k=1}^{N} \frac{\theta_{k0}}{K_k},~\forall t.
\end{equation}
As $t\rightarrow\infty$, \eqref{angle_relation1} becomes
\begin{equation}
\label{angle_relation2}\sum_{k=1}^{N} \frac{{\theta}_k(t \rightarrow \infty)}{K_k} = \sum_{k=1}^{N} \frac{\theta_{k0}}{K_k}.
\end{equation}
On substituting  $\theta_k (t \rightarrow\infty) = \theta_{kf} = \theta_f + ({2(k-1)\pi}/{N}), \forall k$, in \eqref{angle_relation2}, we get \eqref{theta_f}. This completes the proof.
\end{proof}

\begin{remark}
Note that the relation \eqref{angle_relation1} does not tell anything about the individual angular separation between agents in balanced formation. As discussed above, since for $N > 3$, the angular separation between any two consecutive agent's orientations in phase balancing may not be unique, a result, similar to Theorem~\ref{Theorem2}, may not be achievable for $N > 3$. However, the cyclic order of $N$ agents in balanced formation can be assured by using Lemma~\ref{lem1}.
\end{remark}

For the sake of convenience, let us denote by
\begin{equation}
\tilde{\theta}_{k0} = \left(\theta_{k0} - \frac{2(k-1)\pi}{N}\right),
\end{equation}
by using which \eqref{theta_f} can now be compactly written as
\begin{equation}
\label{theta_f_new}\theta_f =  \left(\sum_{k=1}^{N}\frac{\tilde{\theta}_{k0}}{K_k}\right){\Big /}\left(\sum_{k=1}^{N}\frac{1}{K_k}\right).
\end{equation}

Based on this representation, we now state the following corollaries to the Theorem~\ref{Theorem2}.

\begin{cor}\label{cor3}
For the conditions given in Theorem~\ref{Theorem2}, the reference direction $\theta_{f}$, given by \eqref{theta_f_new}, is a convex combination of $\tilde{\theta}_{k0}, \forall k$.
\end{cor}

\begin{proof}
Equation \eqref{theta_f_new} can also be rewritten as
\begin{equation}
\label{theta_f_1}\theta_f =  \sum_{k=1}^{N}\left\{\left(\frac{1}{K_k}\right){\Big /}\left({\sum_{j=1}^{N}\frac{1}{K_j}}\right)\right\}\tilde{\theta}_{k0}.
\end{equation}
Assume that for all $k = 1, \ldots N$,
\begin{equation}
\label{lambda}\lambda_k = \left(\frac{1}{K_k}\right){\Big /}\left(\sum_{j=1}^{N} \frac{1}{K_j}\right).
\end{equation}
Since $K_k > 0$ for all $k = 1, \ldots, N$, $\lambda_k > 0, \forall k$, and satisfies $\sum_{k=1}^{N} \lambda_k = 1$. Substituting \eqref{lambda} in \eqref{theta_f_1}, we get
\begin{equation}
\label{theta_f_2}\theta_f = \sum_{k=1}^{N}\lambda_k\tilde{\theta}_{k0},
\end{equation}
which shows that $\theta_f$ is a convex combination of $\tilde{\theta}_{k0}, \forall k$.
\end{proof}

\begin{cor}\label{cor4}
For the conditions given in Theorem~\ref{Theorem2}, let $\tilde{\theta}_{m0} = \min_{k} \left\{\tilde{\theta}_{k0}\right\}$ and $\tilde{\theta}_{M0} = \max_{k} \left\{\tilde{\theta}_{k0}\right\}$ be the minimum and the maximum angles, respectively. These angles are not reachable as the reference direction in balanced formation of this system of two and three agents.
\end{cor}

\begin{proof}
This can be proved by contradiction. Let us assume that $\tilde{\theta}_{m0}$ is reachable. It means that $\exists~K_k > 0, \forall k$, given according to the Assumption~\ref{assumption}, such that \eqref{theta_f_new} is satisfied. Hence, from \eqref{theta_f_new}, we can write
\begin{equation}
\label{theta_m}\tilde{\theta}_{m0} = \left(\sum_{k=1}^{N}\frac{\tilde{\theta}_{k0}}{K_k}\right){\Big /}\left(\sum_{k=1}^{N}\frac{1}{K_k}\right),
\end{equation}
from which
\begin{equation}
\label{theta_m_1}\displaystyle\sum_{\substack{
   k = 1, \\
   k \neq m
  }}^{N}
\left(\dfrac{\tilde{\theta}_{k0} - \tilde{\theta}_{m0}}{K_k}\right) = 0.
\end{equation}
However, since $\tilde{\theta}_{m0} = \min_{k} \{\tilde{\theta}_{k0}\}$, $ \tilde{\theta}_{k0} -\tilde{\theta}_{m0} > 0 $, for all $k = 1, \ldots, m-1, m+1, \ldots, N$. Thus,
\begin{equation}
\displaystyle\sum_{\substack{
   k = 1, \\
   k \neq m
  }}^{N}
\left(\dfrac{\tilde{\theta}_{k0} - \tilde{\theta}_{m0}}{K_k}\right) > 0
\end{equation}
as $K_k > 0, \forall k$, which contradicts \eqref{theta_m_1}, and hence $\tilde{\theta}_{m0}$ is not reachable as the reference direction. Similarly, we can show that $\tilde{\theta}_{M0}$ is not reachable. This completes the proof.
\end{proof}

Now, we describe the following theorem which gives the reachability of $\theta_f$, defined in \eqref{theta_f_new}, against heterogeneous control gains $K_k, \forall k$, given according to the Assumption~\ref{assumption}.

\begin{thm}\label{Theorem3}
Consider $N \in \{2, 3\}$ agents, with dynamics given by \eqref{modelNew}, under the control law \eqref{control3_1} with non-zero heterogeneous control gains $K_k, \forall k$, given according to the Assumption~\ref{assumption}. Let the initial heading angles of the agents be given by $\pmb{\theta}(0)$ such that the initial velocity vectors $e^{i\theta_{10}}, \ldots, e^{i\theta_{N0}}$, are in cyclic order with  $\theta_{10} < \theta_{20} < \ldots < \theta_{N0}$. Let the orientations of agents in balanced formation be given by $\theta_{kf} = \theta_f + ({2(k-1)\pi}/{N}), \forall k$. Then, the reference direction $\theta_f$, given by \eqref{theta_f_new}, is reachable if and only if
\begin{equation}
\theta_f \in (\tilde{\theta}_{m0}, \tilde{\theta}_{M0}),
\end{equation}
where, $\tilde{\theta}_{m0}$, and $\tilde{\theta}_{M0}$, are defined in Corollary~\ref{cor4}.
\end{thm}

\begin{proof}
This directly follows from Corollary~\ref{cor3} and Corollary~\ref{cor4} that $\theta_f \in (\tilde{\theta}_{m0}, \tilde{\theta}_{M0})$, depending upon the heterogeneous controller gains $K_k, \forall k$, given according to the Assumption~\ref{assumption}. The sufficiency condition is proved as follows.

Let $\theta_f \in (\tilde{\theta}_{m0}, \tilde{\theta}_{M0})$. Then, we can find $\sigma_k$ for all $k = 1, \ldots, N$, such that
\begin{equation}
\label{rel_4}\sum_{k=1}^{N}\sigma_k \hat{\theta}_{k0} = \theta_f,
\end{equation}
where, $\sum_{k=1}^{N} \sigma_k = 1$, with $\sigma_k > 0, \forall k$. Let us define
\begin{equation}
\label{gain}K_k = c/\sigma_k,
\end{equation}
for all $k$, where, $c > 0$, is a constant. Thus, $K_k > 0, \forall k$, and satisfies $\sum_{k=1}^{N} (1/K_k) = 1/c$. Moreover, for the given value of $c$, the parameter $\sigma_k, \forall k$, can be chosen appropriately so that the heterogeneous gains $K_k, \forall k$, defined in \eqref{gain}, satisfy Assumption~\ref{assumption}. Replacing $\sigma_k = c/K_k$ in \eqref{rel_4}, we get
\begin{equation}
\nonumber {\theta}_f = \sum_{k=1}^{N} \left\{\dfrac{\dfrac{1}{K_k}}{\dfrac{1}{c}}\right\}\tilde{\theta}_{k0} =\sum_{k=1}^{N} \left\{\left(\dfrac{\dfrac{1}{K_k}}{\displaystyle\sum_{j=1}^{N}\dfrac{1}{K_j}}\right)\tilde{\theta}_{k0}\right\} = \dfrac{\displaystyle\sum_{k=1}^{N} \dfrac{\tilde{\theta}_{k0}}{K_k}}{\displaystyle\sum_{k=1}^{N} \dfrac{1}{K_k}},
\end{equation}
which is the same as \eqref{theta_f_new}. This completes the proof.
\end{proof}

\begin{remark}
If we choose homogeneous controller gains, as in \cite{Sepulchre2007}, that is, $K_k = K > 0,~\forall k$, then the reference direction ${\theta}_f$, by using \eqref{theta_f_new}, is given by
\begin{equation}
\label{theta_f_avg}\overline{\theta}_f = \frac{1}{N}\sum_{k=1}^{N}\tilde{\theta}_{k0},
\end{equation}
which is the average of all $\tilde{\theta}_{k0}, \forall k$. Thus, for the given initial heading angles, a unique reference direction $\overline{\theta}_f$, given by \eqref{theta_f_avg}, is possible by using homogeneous controller gains, and hence a unique arrangement of agents' velocity vectors is possible in phase balancing. However, by using heterogeneous controller gains, we are able to expand the reachable set of the reference direction $\theta_f$. In fact, the agents can be made to converge to any desired reference direction $\theta_f \in (\tilde{\theta}_{m0}, \tilde{\theta}_{M0})$ by suitably selecting the heterogeneous gains $K_k > 0, \forall k$, given according to the Assumption~\ref{assumption}. The heterogeneous gains can be selected according to \eqref{gain}. We can see that these gains are not unique since none of $\sigma_k$, $c$ need to be unique. We also observe that \eqref{theta_f_new} is independent of the initial locations of the agents. Therefore, different groups of the agents, with arbitrary initial locations, but with same individual initial velocity directions, can be made to converge in balanced formation with the same desired reference direction $\theta_f$.
\end{remark}

Since it is physically impossible to get the same gains for all the agents, the idea of heterogeneous controller gains was introduced. Suppose the homogeneous gains $K$ of the agent vary within certain limits while obeying all the conditions of convergence, then we have the following theorem, which tells about the deviation of the reference direction $\theta_f$ from its mean value $\overline{\theta}_f$, given by \eqref{theta_f_avg}, and comments on its reachability.

For the sake of clarity, we state the theorem for some restricted set of initial heading angles, which can easily be extended for the general setting of initial heading angles as discussed in the remark below the theorem.

\begin{thm}\label{Theorem4}
Consider $N \in \{2, 3\}$ agents, with dynamics given by \eqref{modelNew}, under the control law \eqref{control3_1} with homogeneous control gains $K_k = K > 0, \forall k$.  Let there be an error of $\epsilon_k = \sigma_k K$, where $0 \leq \sigma_k < 1$, in the gain $K$ of the $k^\text{th}$ agent, such that the erroneous gains $K \pm \epsilon_k, \forall k$, obey the Assumption~\ref{assumption}. Let $\sigma = \max_k\{\sigma_k\}$ be the maximum error, and the initial heading angles of the agents be given by $\pmb{\theta}(0)$ with  $\theta_{10} < \theta_{20} < \ldots < \theta_{N0}$ such that $\tilde{\theta}_{k0}, \forall k$, are non-positive. Then, in balanced formation of this system of two and three agents, the perturbed reference direction, $\theta^p_f$, is contained in
\begin{equation}
\label{error_angle}{\theta}^p_f \in \left(\tilde{\theta}_{m0}, \tilde{\theta}_{M0}\right) \bigcap \left[\overline{{\theta}}_f - \Delta{\check{\theta}}^p_f, \overline{{\theta}}_f + \Delta{\hat{\theta}}^p_f\right],
\end{equation}
where,
\begin{equation}
\Delta{\check{\theta}}^p_f =  -\left(\frac{2\sigma}{1 - \sigma}\right)\overline{{\theta}}_f,~~\text{and}~~\Delta{\hat{\theta}}^p_f =  -\left(\frac{2\sigma}{1 + \sigma}\right)\overline{{\theta}}_f,
\end{equation}
are, respectively, the maximum values of the lower and upper deviations of the reachable velocity direction from its mean value $\overline{{\theta}}_f$, given by \eqref{theta_f_avg}.
\end{thm}

\begin{proof}
Since the erroneous controller gain of the $k^\text{th}$ agent is $K \pm \epsilon_k$, by using \eqref{theta_f_1}, we can write
\begin{equation}
\label{error1}{\theta}^p_f =  \sum_{k=1}^{N}\left\{\left(\frac{1}{K \pm \epsilon_k}\right){\Big /}\left({\sum_{j=1}^{N}\frac{1}{K \pm \epsilon_j}}\right)\right\}\tilde{\theta}_{k0}.
\end{equation}
Since $\tilde{\theta}_{k0}, \forall k$, are non-positive as per our consideration, the lower bound of $\tilde{\theta}^p_f$, denoted by $\check{\theta}^p_f$, is given by
\begin{equation}
\label{error2}\check{\theta}^p_f =  \sum_{k=1}^{N}\left\{\left(\frac{1}{K - \epsilon_k}\right){\Big /}\left({\sum_{j=1}^{N}\frac{1}{K + \epsilon_j}}\right)\right\}\tilde{\theta}_{k0}.
\end{equation}
Substituting $\epsilon_k = \sigma_k K$, in \eqref{error2}, we get
\begin{eqnarray}
\check{\theta}^p_f  &=& \sum_{k=1}^{N}\left\{\left(\frac{1}{1 - \sigma_k}\right){\Big /}\left({\sum_{j=1}^{N}\frac{1}{1 + \sigma_j}}\right)\right\}\tilde{\theta}_{k0}\\
& \geq & \sum_{k=1}^{N}\left\{\left(\frac{1}{1 - \sigma}\right){\Big /}\left(\frac{N}{1 + \sigma}\right)\right\}\tilde{\theta}_{k0}=  \left(\frac{1 + \sigma}{1 - \sigma}\right)\overline{{\theta}}_f.
\end{eqnarray}
Similarly, the upper bound of ${\theta}^p_f$ is given by
\begin{eqnarray}
\hat{\theta}^p_f  &=& \sum_{k=1}^{N}\left\{\left(\frac{1}{1 + \sigma_k}\right){\Big /}\left({\sum_{j=1}^{N}\frac{1}{1 - \sigma_j}}\right)\right\}\tilde{\theta}_{k0}\\
& \geq & \sum_{k=1}^{N}\left\{\left(\frac{1}{1 + \sigma}\right){\Big /}\left(\frac{N}{1 - \sigma}\right)\right\}\tilde{\theta}_{k0}=  \left(\frac{1 - \sigma}{1 + \sigma}\right)\overline{{\theta}}_f.
\end{eqnarray}
Thus, the maximum values of the lower and upper deviations of ${\theta}^p_f$ from its mean value $\overline{{\theta}}_f$ are, respectively,
\begin{eqnarray}
\Delta{\check{\theta}}^p_f &=&  \overline{{\theta}}_f - \left(\frac{1 + \sigma}{1 - \sigma}\right)\overline{{\theta}}_f = -\left(\frac{2\sigma}{1 - \sigma}\right)\overline{{\theta}}_f\\
\Delta{\hat{\theta}}^u_f &=&  \left(\frac{1 - \sigma}{1 + \sigma}\right)\overline{{\theta}}_f - \overline{{\theta}}_f = -\left(\frac{2\sigma}{1 + \sigma}\right)\overline{{\theta}}_f.
\end{eqnarray}
It follows from the above discussion that
\begin{equation}
{\theta}^p_f \in \left[\overline{{\theta}}_f - \Delta{\check{\theta}}^p_f, \overline{{\theta}}_f + \Delta{\hat{\theta}}^p_f\right].
\end{equation}
However, since Theorem~\ref{Theorem3} ensures that ${\theta}^p_f \in (\tilde{\theta}_{m0}, \tilde{\theta}_{M0})$ when there is heterogeneity in the controller gains, the actual set of angles reachable by ${\theta}^p_f$ is \eqref{error_angle}. This completes the proof.
\end{proof}

\begin{remark}
The lower and upper bounds of the perturbed reference direction ${\theta}^p_f$, when $\tilde{\theta}_{k0}$, associated to the $k^\text{th}$ agent, is either positive or negative, can be obtained in the same manner as in Theorem~\ref{Theorem4} by appropriately minimizing or maximizing the coefficients of $\tilde\theta_{k0}, \forall k$, depending upon its sign. Thus, these cases are not presented here to avoid repetition.
\end{remark}

\subsection{Case~2: $\omega_0 \neq 0$}
In this case, the motion of each agent is governed by \eqref{control3}. As a result, the agents move around their individual circular orbits at an angular frequency $\omega_0$ in balanced formation. For ease of analysis in this framework, it is convenient to use a frame of reference that rotates at the same frequency $\omega_0$. Thus, by replacing $\theta_k \rightarrow \theta_k + \omega_0 t$ in \eqref{control3}, which corresponds to a rotating frame at frequency $\omega_0$, we get the turn rate of the $k^\text{th}$ agent as
\begin{equation}
\dot{\theta}_k = -\frac{K_k}{N}\sum_{j=1}^{N}\sin(\theta_j - \theta_k),
\end{equation}
which is the same as \eqref{control3_1}. Therefore, all the analysis remains unchanged in a rotating frame of reference, and hence omitted.

The usefulness of heterogeneous gains for this case lies in the following aspect. Note that the center of the circular orbit traversed by the $k^\text{th}$ agent is given by
\begin{equation}
c_k = r_k + i\omega^{-1}_0e^{i\theta_k}.
\end{equation}
As discussed above, since the heading angles $\theta_k$ of the agents depend on the heterogeneous gains, the center $c_k$ of the individual circular orbit depends on the heterogeneous control gains. Thus, various circular orbits of the motion of agents in balanced formation can be obtained by using heterogeneous control gains, and hence, the area of interest can be explored more effectively.

\begin{figure*}
\centering
\subfigure[]{\includegraphics[scale=0.4]{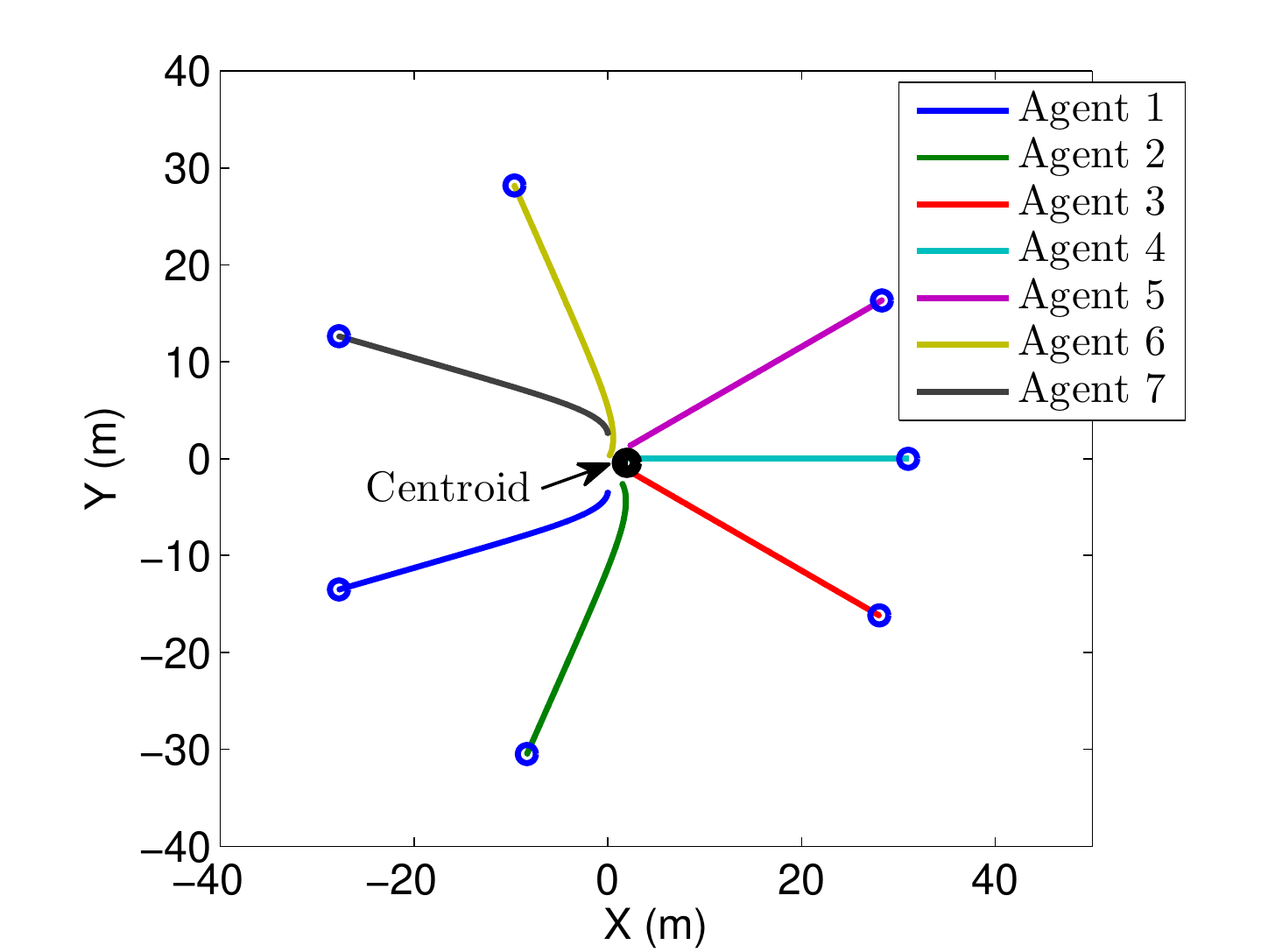}}
\subfigure[]{\includegraphics[scale=0.4]{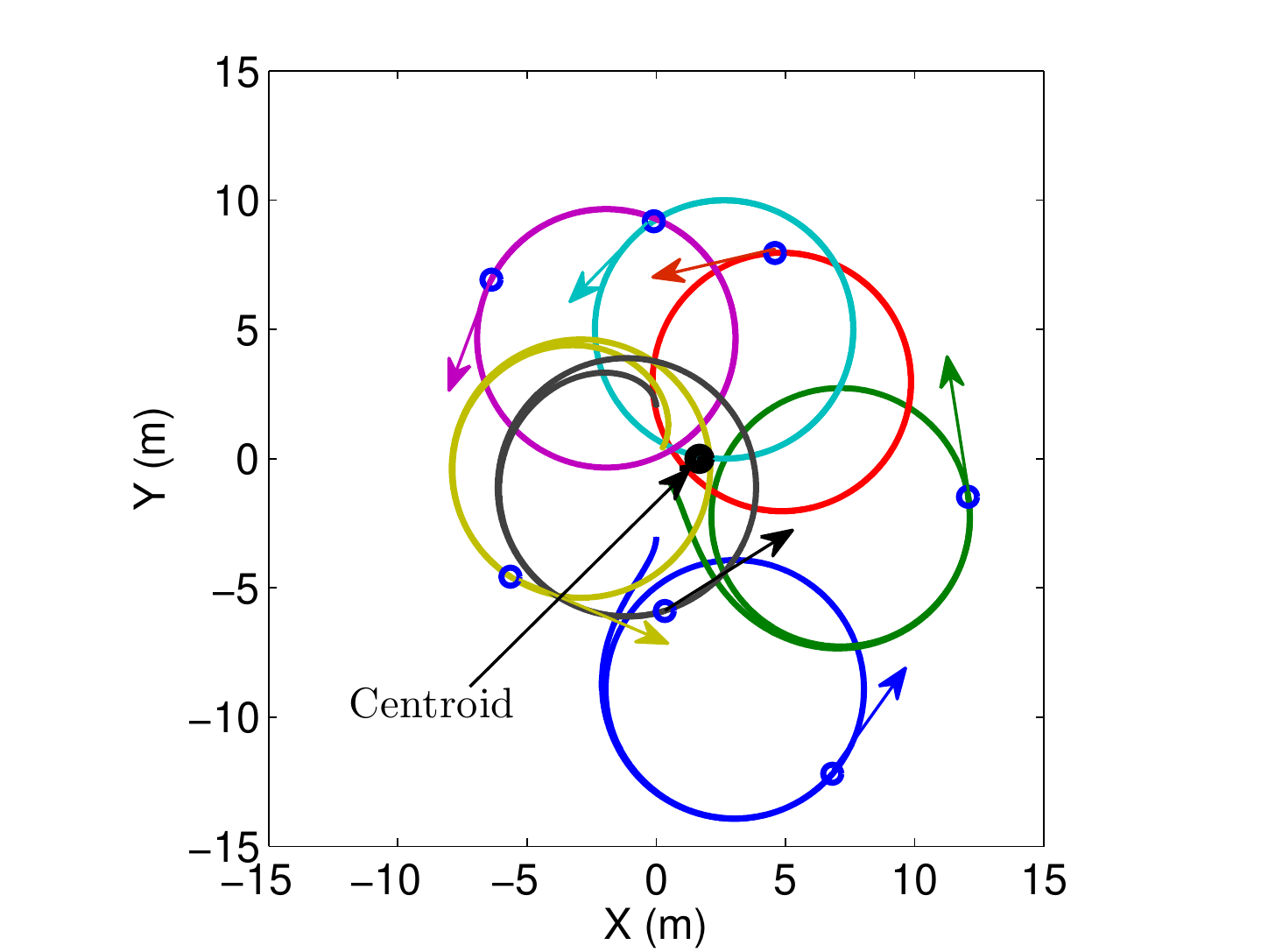}}
\subfigure[]{\includegraphics[scale=0.4]{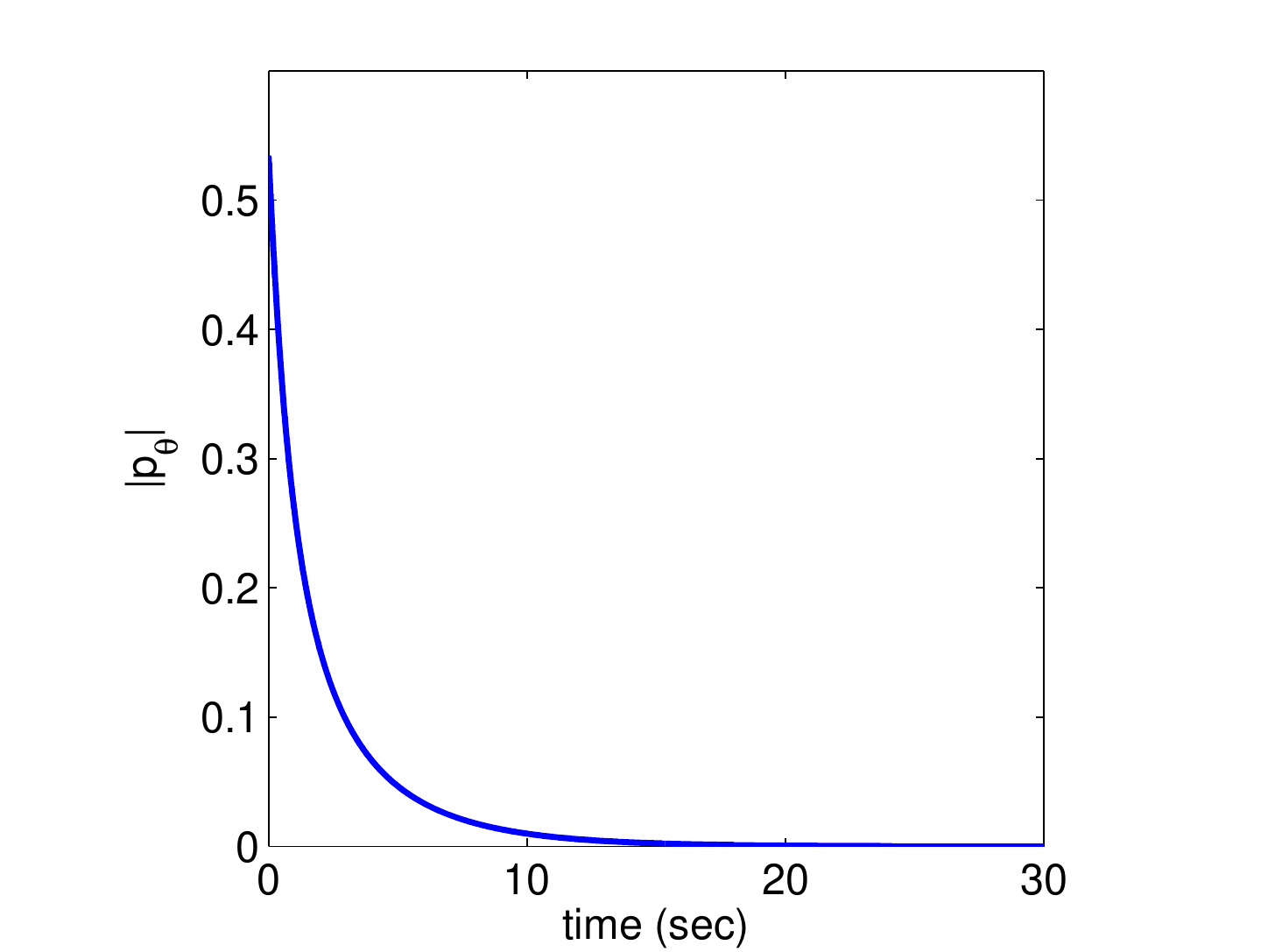}}
\caption{Balanced formation of seven agents under the control law \eqref{control3} for heterogeneous gains $K_{set1} = \{2, 1, 0, 0, 0, 1, 2\}$. Trajectories of the agents for $(a)$ $\omega_0 = 0$. $(b)$ $\omega_0 = 0.2~\text{rad/sec}$. $(c)$ Convergence of the phase order parameter $p_\theta$ to zero with time for $\omega_0 = 0$.}
\label{balanced formation seven agents}
\end{figure*}

\begin{figure*}
\centering
\subfigure[]{\includegraphics[scale=0.4]{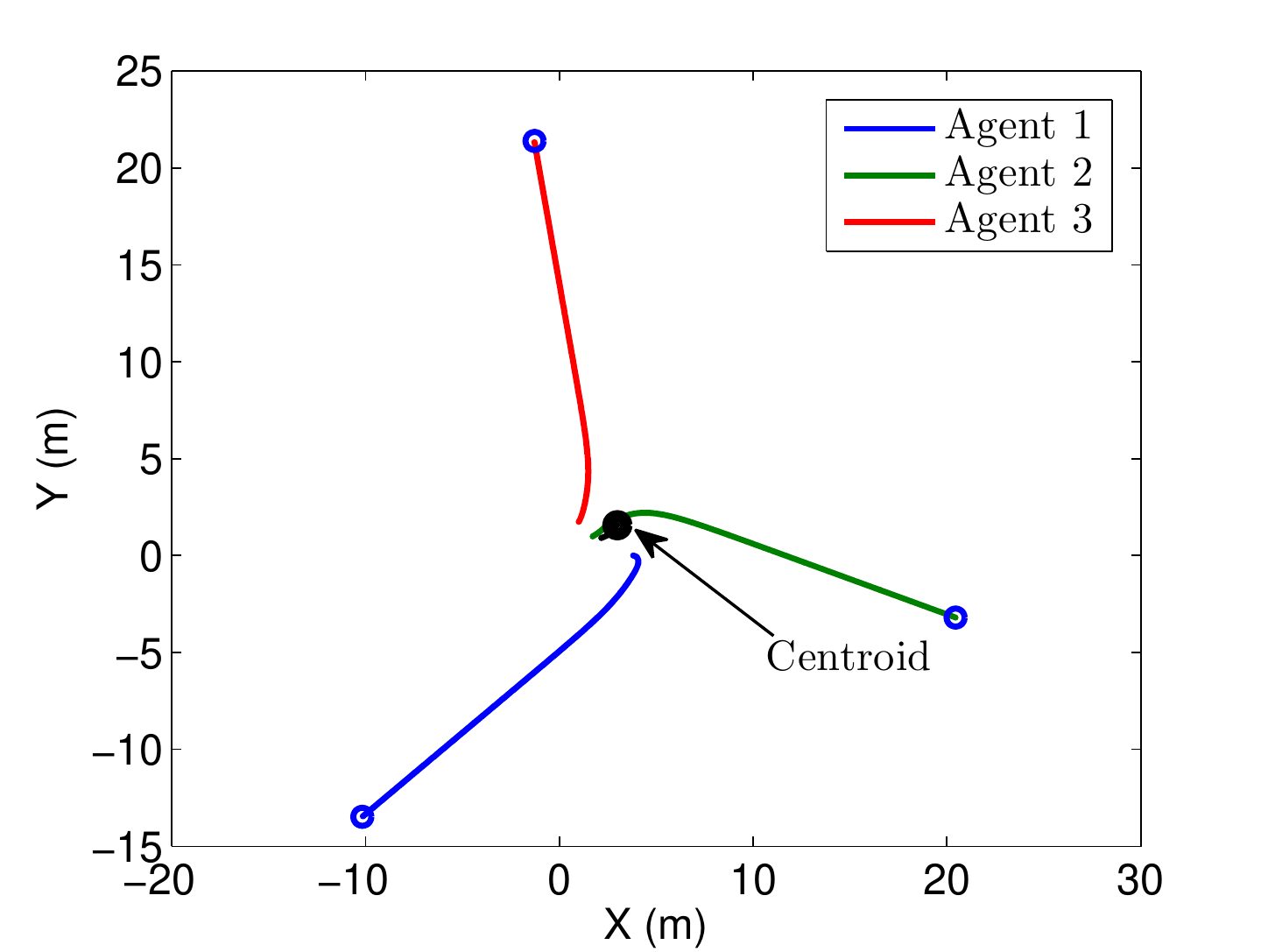}}\hspace{1cm}
\subfigure[]{\includegraphics[scale=0.4]{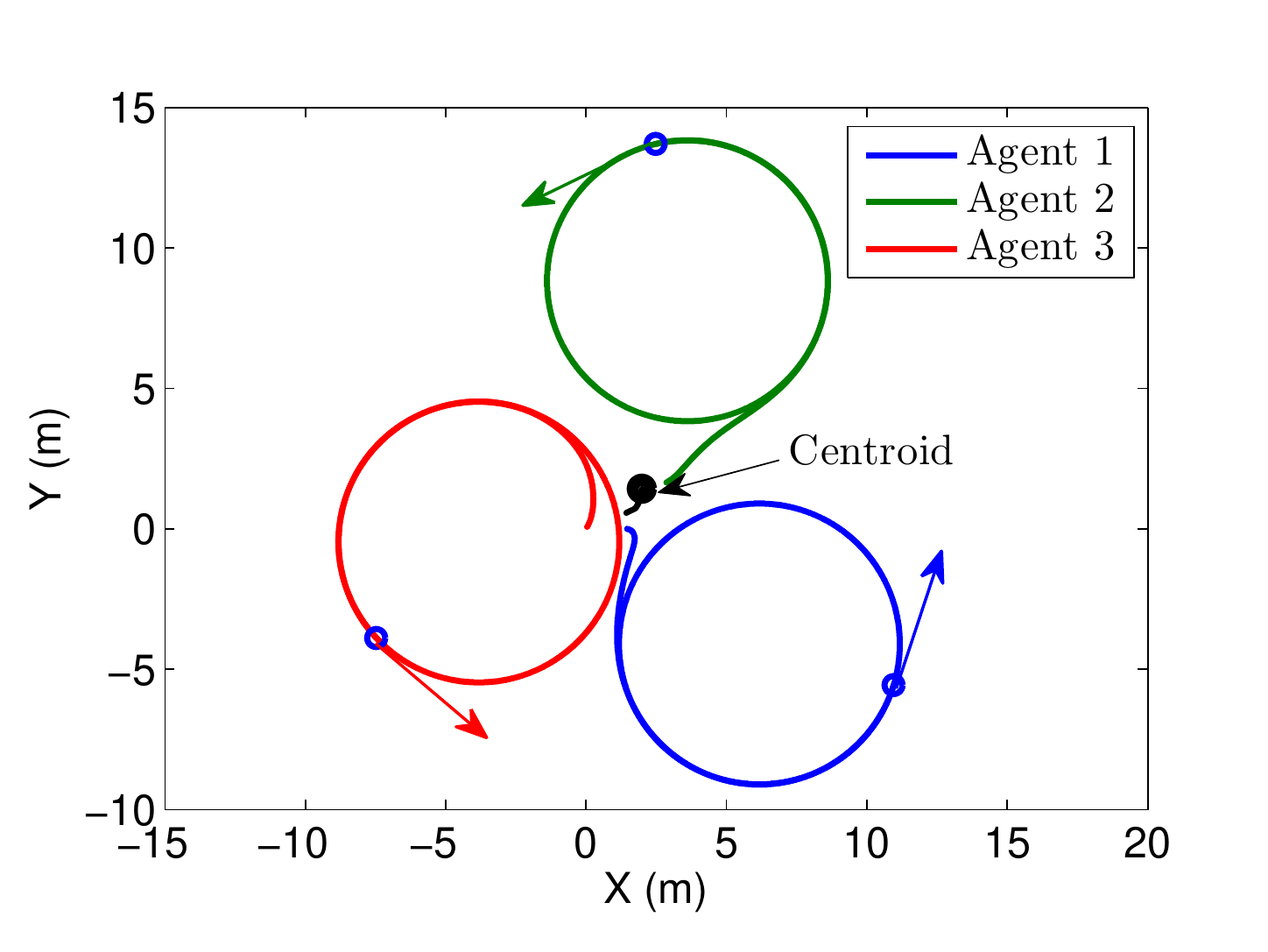}}\hspace{1cm}
\subfigure[]{\includegraphics[scale=0.4]{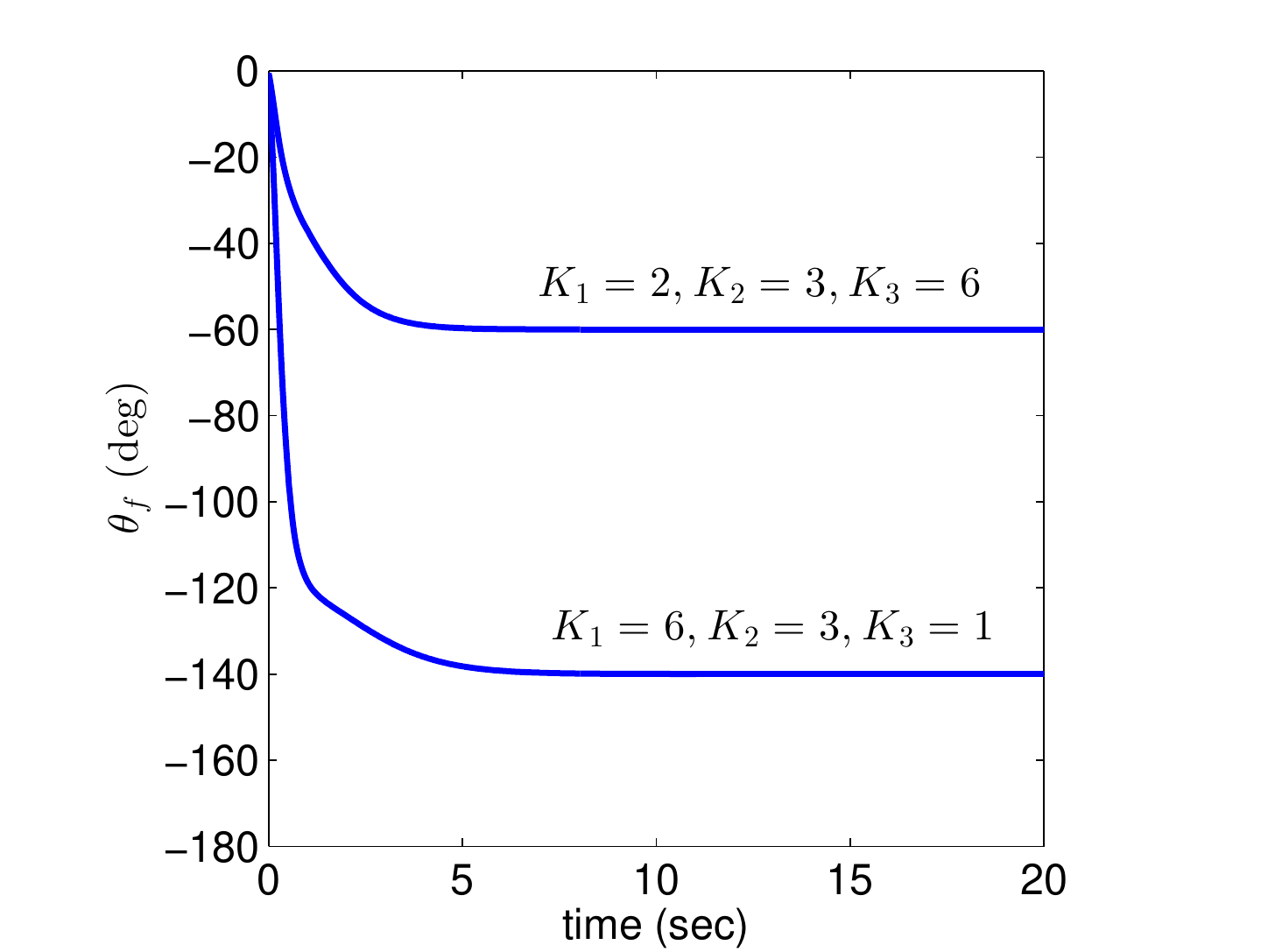}}\hspace{1cm}
\subfigure[]{\includegraphics[scale=0.4]{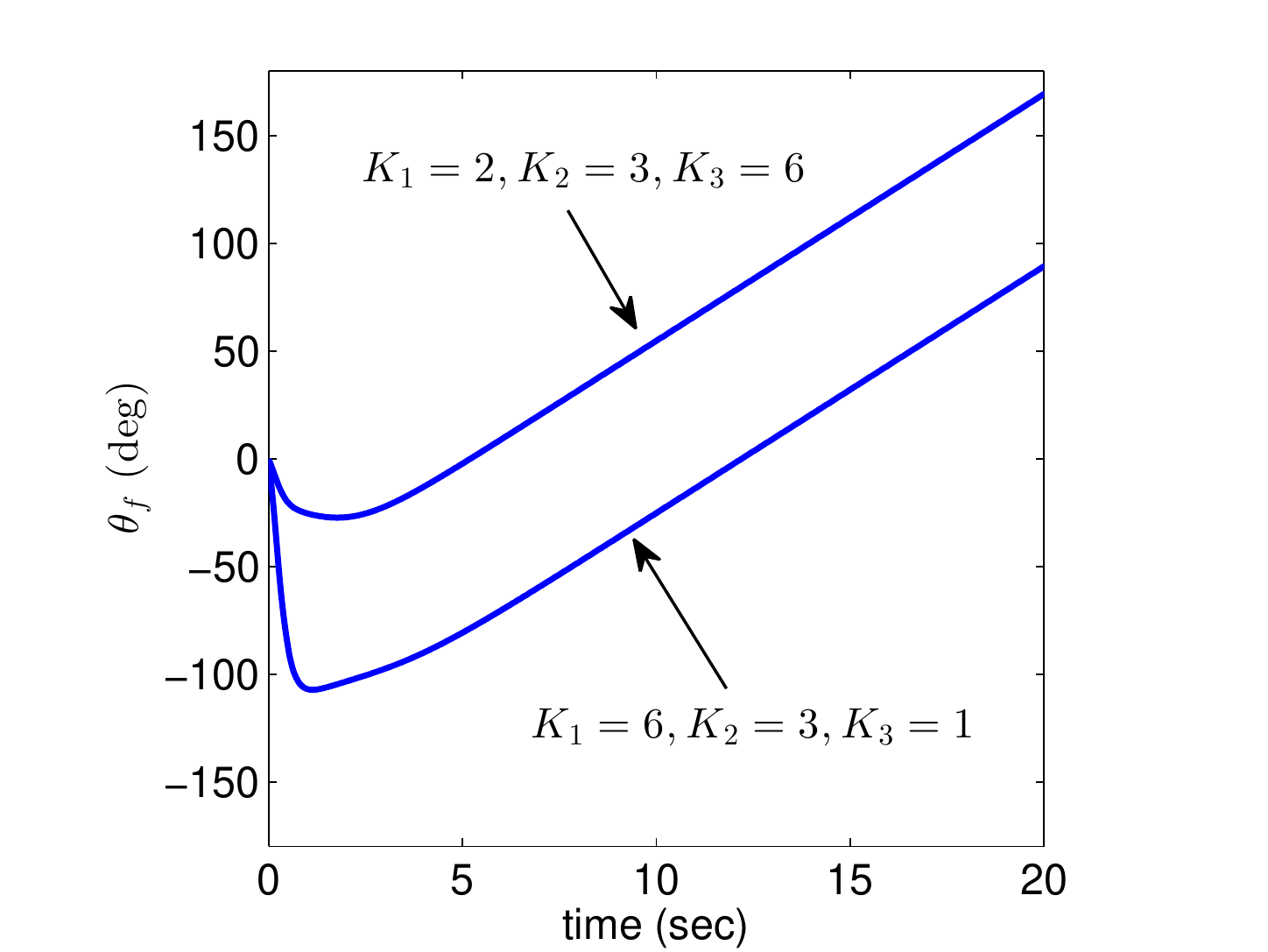}}
\caption{Balanced formation of three agents under the control law \eqref{control3} with the two set of gains $K_{set2} = \{2, 3, 6\}$ and $K_{set3} = \{6, 3, 1\}$. Trajectories of the agents when $(a)$ $\omega_0 = 0$. $(b)$ $\omega_0 = 0.2~\text{rad/sec}$. Variations in the reference direction $\theta_f$ with time for the two sets of gains when $(c)$ $\omega_0 = 0$. $(d)$ $\omega_0 = 0.2~\text{rad/sec}$. Note that the reference direction $\theta_f$ is actually the velocity direction of agent\#1 as per our convention.}
\label{balanced formation three agents}
\end{figure*}

\begin{example}\label{ex1}
Consider $N=7$ agents starting from arbitrary initial positions with initial heading angles $\pmb{\theta}(0) = [-90^\circ, -60^\circ, -30^\circ, 0^\circ, 30^\circ, 60^\circ, 90^\circ]^T$. One can observe that the initial phase order parameter vector $p_{\theta_0}$ of these agents lies along the real axis. According to Lemma~\ref{lem1}, consider the heterogeneous gains $K_{set1} = \{K_k, k = 1, \ldots N\} = \{2, 1, 0, 0, 0, 1, 2\}$. In Fig.~\ref{balanced formation seven agents}, balanced formation of the agents for $K_{set1}$ under control \eqref{control3}, is shown for both $\omega_0 = 0$ and $\omega_0 = 0.2~\text{rad/sec}$. The trajectories of the agents are shown in Figs.~$3(a)$ and $3(b)$, while the variation of phase order parameter $p_\theta$ with time is shown in Fig.~$3(c)$ for $\omega_0 = 0$, and is similar for $\omega_0 = 0.2~\text{rad/sec}$, hence, this is not shown. Note that phase balancing is achievable if the heterogeneous gains are zero for $\lfloor N/2 \rfloor$ agents.
\end{example}

\begin{example}\label{ex2}
In this example, we consider three agents starting from arbitrary initial positions with initial heading angles $\pmb{\theta}(0) = [0^\circ, 30^\circ, 60^\circ]^T$. In Fig.~\ref{balanced formation three agents}, balanced formation for the two set of gains $K_{set2} = \{2, 3, 6\}$ and $K_{set3} = \{6, 3, 1\}$ is shown under the control \eqref{control3} for both $\omega_0 = 0$ and $\omega_0 = 0.2~\text{rad/sec}$. In Figs.~$4(a)$ and $4(b)$, the trajectories of the agents are shown only for $K_{set3}$, and are similar for $K_{set2}$ while the corresponding variation in the reference direction $\theta_f$ in time is shown in Fig.~$4(c)$ for the two sets of gains. Here, since $\tilde{\theta}_{m0} = -180^\circ$ and $\tilde{\theta}_{M0} = 0^\circ$, it follows from Theorem~\ref{Theorem3} that the reference direction $\theta_f \in (-180^\circ, 0)$ when $\omega_0 = 0$. Note that, in the case when $\omega_0 \neq 0$, since the agents continue to rotate around individual circles in balanced formation, the reference direction $\theta_f$ (which is the velocity direction of the agent\#1) keeps increasing with time.
\end{example}


\section{Two Agents: A Few Interesting Results}
In this section, we address the special case of two agents and show that, unlike $K_k > 0, \forall k$, their exists a less restrictive condition on the heterogeneous gains $K_k$, which results in further expansion of the reachable set of the reference direction of the agents in balanced formation. We present the results only for $\omega_0 = 0$ since the analysis is unchanged for $\omega_0 \neq 0$ in a rotating frame of reference by redefining $\theta_k \rightarrow \theta_k + \omega_0 t$ for the $k^\text{th}$ agent.

\subsection{Analysis of Heterogeneous Controller Gains}
For $N=2$, the time derivative of the potential function ${U}(\pmb{\theta})$ from \eqref{U_dot} is given by
\begin{equation}
\label{U_dot_two_agents}\dot{U}(\pmb{\theta}) \big{|}_{N=2} = -\frac{1}{2^2}(K_1 + K_2) \sin^2(\theta_2-\theta_1),
\end{equation}
which implies that the potential ${U}(\pmb{\theta})$ is decreasing if $K_1 + K_2 > 0$ since $\sin^2(\theta_2-\theta_1) > 0$. Moreover, it is easy to verify that $\sin^2(\theta_2-\theta_1) = 0$, only for the trivial cases when both the agents are already synchronized or balanced. Using Theorem~\ref{Theorem1}, it follows from \eqref{U_dot_two_agents} that if $K_1 + K_2 > 0$, agents asymptotically stabilize to a balanced formation. Hence, phase balancing of the agents is achievable for both positive and negative values of gains $K_1$ and $K_2$ provided that $K_1 + K_2 > 0$.


\begin{figure}[!t]
\centering
\includegraphics[scale=0.4]{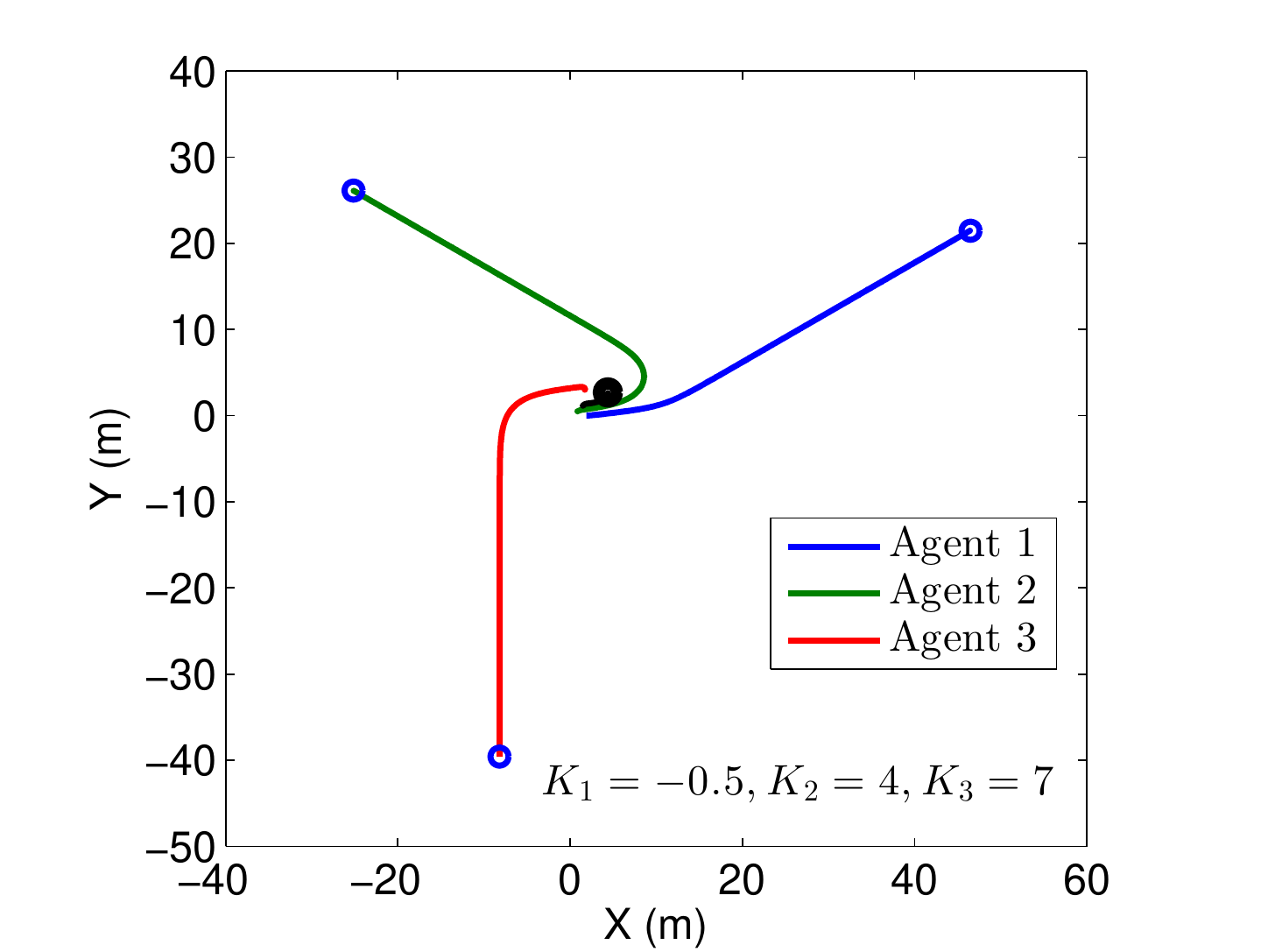}
\caption{Balanced formation of three agents under the control \eqref{control3_1} with heterogeneous gains ${K}_\text{set4} = \{-0.5, 4, 7\}$. The steady-state orientation of the agent\#1 is $\theta_f = 30^\circ$, which is outside the interval $(\tilde{\theta}_{m0}, \tilde{\theta}_{M0}) = (-180^\circ, 0)$.}
\label{balanced formation three agents with negative gain}
\end{figure}


\begin{remark}
For $N>2$, we did not come up with a simplified expression for the sufficient condition on the controller gains $K_k$, however, simulation results show that their exists a combination of both positive and negative values of the controller gains $K_k$ that gives rise to a balanced formation with an extended set of the reference direction ${\theta}_f$. For example, the reference direction (that is, the velocity direction of agent\#1) in balanced formation of the three-agent system considered in Example~\ref{ex2} lies outside the interval $(-180^\circ, 0)$ as shown in Fig.~\ref{balanced formation three agents with negative gain} for the set of gains ${K}_\text{set4} = \{-0.5, 4, 7\}$.
\end{remark}

\subsection{Reachable Velocity Directions}
In this subsections, we describe a theorem, which says that the reachable set of the reference direction $\theta_f$, given by \eqref{theta_f_new}, further expands when both positive and negative values of gains $K_1$ and $K_2$, satisfying $K_1 + K_2 > 0$, are selected.

Without loss of generality and for the sake of clarity, we consider that the agents start with initial headings $\theta_{10} (= 0) < \theta_{20} < \pi$, which can be ensured by via a rotation of the original coordinate system by an angle $\theta_R \in (-\pi, \pi)$, which is chosen such that the real axis of this new coordinate system lies along that initial unit vector $e^{i\theta_{k0}}, k =1, 2$, which ensures that both the initial heading angles $\theta_{10}, \theta_{20}$, in this new coordinate system, are non-negative (measured anti-clockwise from the new real-axis). Based on this, the following theorem is now stated.

\begin{thm}\label{Theorem5}
Consider two agents, with dynamics given by \eqref{modelNew}, under the control law \eqref{control3}. Let the initial heading angles of the agents be given by $\theta_{10} (= 0) < \theta_{20} < \pi$. Then, any ${\theta}_f \in [-\pi, \pi]$, which is the reference direction of this system of two agents in balanced formation, is reachable if and only if there exist controller gains $K_1$ and $K_2$ such that $K_1 + K_2 > 0$.
\end{thm}

\begin{proof}
For $N=2$, the reference direction ${\theta}_f$, by using \eqref{theta_f_new}, is given by
\begin{equation}
\label{theta_f_two_agents}{\theta}_f =  \left(\dfrac{K_2}{K_1 + K_2}\right) \tilde{\theta}_{10} + \left(\dfrac{K_1}{K_1 + K_2}\right) \tilde{\theta}_{20}
\end{equation}
Substituting
\begin{equation}
\label{lamda1}\lambda_1 = \left(\dfrac{K_2}{K_1 + K_2}\right)~~ \text{and}~~\lambda_2 = \left(\dfrac{K_1}{K_1 + K_2}\right)
\end{equation}
in \eqref{theta_f_two_agents}, we get
\begin{equation}
\label{theta_f_two_agents_1}{\theta}_f = \lambda_1\tilde{\theta}_{10} + \lambda_2\tilde{\theta}_{20}.
\end{equation}
Note that the parameters $\lambda_1$ and $\lambda_2$ satisfy $\lambda_1 + \lambda_2 = 1$. As per our consideration of initial velocity directions of the agents, here, it holds that $\tilde{\theta}_{m0} = \tilde{\theta}_{20}$ and $\tilde{\theta}_{M0} = \tilde{\theta}_{10}$, where $\tilde{\theta}_{m0}$, and $\tilde{\theta}_{M0}$, are defined in Corollory~\ref{cor4}. Now, depending upon the various choices of gains $K_1$ and $K_2$ satisfying $K_1 + K_2 > 0$, we consider the following three cases.

{\it Case~$1$}:
Let us assume that the gains $K_1 > 0$ and $K_2 > 0$. It implies that $\lambda_1 > 0$ and $\lambda_2 > 0$. In this situation, the proof directly follows from Theorem~\ref{Theorem3}, which ensures that ${\theta}_f$ is reachable iff
\begin{equation}
{\theta}_f \in (\tilde{\theta}_{m0}, \tilde{\theta}_{M0}).
\end{equation}
Substituting for $\tilde{\theta}_{m0}$, and $\tilde{\theta}_{M0}$, we have
\begin{equation}
{\theta}_f \in (\tilde{\theta}_{20}, \tilde{\theta}_{10}).
\end{equation}

{\it Case~$2$}:
Assume that the gains $K_1 > 0$, $K_2 \leq 0$ and satisfy $K_1 + K_2 > 0$. It implies that $\lambda_1 \leq 0$ and $\lambda_2 > 0$. Thus, by using relation $\lambda_2 = 1 - \lambda_1$, \eqref{theta_f_two_agents_1} can be written as
\begin{equation}
\label{relation3_1}{\theta}_f - \tilde{\theta}_{20} = \lambda_1(\tilde{\theta}_{10} - \tilde{\theta}_{20}).
\end{equation}
RHS (right-hand side) of \eqref{relation3_1} is non-positive, that is, $\lambda_1(\tilde{\theta}_{10} - \tilde{\theta}_{20}) \leq 0$ since $\lambda_1 \leq 0$ and $\tilde{\theta}_{10} > \tilde{\theta}_{20}$. Therefore, LHS (left-hand side) of \eqref{relation3_1} should also be non-positive, that is,
\begin{equation}
-\pi \leq {\theta}_f \leq \tilde{\theta}_{20}.
\end{equation}

{\it Case~$3$}:
Now, let us assume that the gains $K_1 \leq 0$, $K_2 > 0$ and satisfy $K_1 + K_2 > 0$. It implies that $\lambda_1 > 0$ and $\lambda_2 \leq 0$. Thus, by using relation $\lambda_1 = 1 - \lambda_2$, \eqref{theta_f_two_agents_1} can be written as
\begin{equation}
\label{relation4}{\theta}_f - \tilde{\theta}_{10} = -\lambda_2(\tilde{\theta}_{10} - \tilde{\theta}_{20}).
\end{equation}
RHS of \eqref{relation4} is non-negative, that is, $-\lambda_2(\tilde{\theta}_{10} - \tilde{\theta}_{20}) \geq 0$ since $\lambda_2 \leq 0$ and $\tilde{\theta}_{10} > \tilde{\theta}_{20}$. Therefore, LHS of \eqref{relation4} should also be non-negative, that is,
\begin{equation}
\tilde{\theta}_{10} \leq {\theta}_f \leq \pi.
\end{equation}

All the above cases lead to the conclusion that ${\theta}_f \in [-\pi, \pi]$. This proves the necessary condition. To prove sufficiency condition for these two cases, we again consider the following cases.

{\it Case~$1$}:
Let $-\pi \leq {\theta}_f \leq \tilde{\theta}_{20}$ is reachable. Then according to \eqref{relation3_1}, the angular difference ${\theta}_f - \tilde{\theta}_{20}$ can be expressed as
\begin{equation}
\label{relation5}{\theta}_f - \tilde{\theta}_{20} = -\alpha(\tilde{\theta}_{10} - \tilde{\theta}_{20})
\end{equation}
where, $ \alpha \geq 0$. Let us define $K_1 = (1 +\alpha)/c$ and $K_2 = -\alpha/c$, where $c > 0$ is a constant. Thus, $K_1 > 0$ and $K_2 \leq 0$ and satisfy $K_1 + K_2 = ({1}/{c})$.

Replacing $(1+\alpha)$ and $\alpha$ by $cK_1$ and $-cK_2$, respectively, in \eqref{relation5}, we get
\begin{equation}
{\theta}_f = \left(\dfrac{K_2}{K_1 + K_2}\right)\tilde{\theta}_{10} + \left(\dfrac{K_1}{K_1 + K_2}\right)\tilde{\theta}_{20},
\end{equation}
which is the same as \eqref{theta_f_two_agents}.

{\it Case~$2$}: Let $ \tilde{\theta}_{10} \leq {\theta}_f \leq \pi$ is reachable. Then, according to \eqref{relation4}, the angular difference ${\theta}_f - \tilde{\theta}_{10}$ can be expressed as
\begin{equation}
\label{relation6}{\theta}_f - \tilde{\theta}_{10} = \beta(\tilde{\theta}_{10} - \tilde{\theta}_{20})
\end{equation}
where, $ \beta \geq 0$. Let us define $K_1 = -\beta/c$ and $K_2 = (1 + \beta)/c$, where $c > 0$ is a constant. Thus, $K_1 \leq 0$ and $K_2 > 0$ and again satisfy $K_1 + K_2 = ({1}/{c})$.

Replacing $\beta$ and $(1+\beta)$ by ${-c}{K_1}$ and ${c}{K_2}$, respectively in \eqref{relation6}, we again get \eqref{theta_f_two_agents}. These results imply that the phase balancing of the agents can be achieved at any desired reference direction $\theta_f \in [-\pi, \pi]$ for the suitable choices of controller gains $K_1$ and $K_2$ provided $K_1 + K_2 > 0$. This completes the proof.
\end{proof}

Pictorially, Theorem~\ref{Theorem5} is summarized in Fig.~\ref{Pictorial representation of Theorem}.


\begin{figure}[!t]
\centering
\includegraphics[scale=0.65]{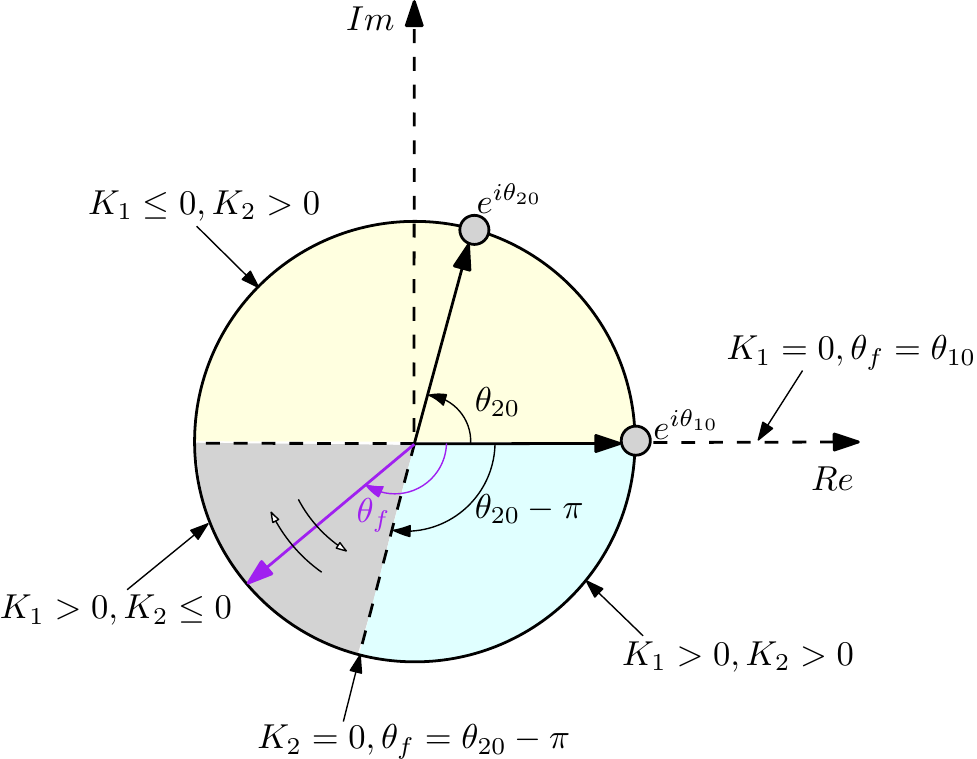}
\caption{Pictorial representation of Theorem~\ref{Theorem5}.}
\label{Pictorial representation of Theorem}
\end{figure}



\begin{figure}[!t]
\centering
\includegraphics[scale=0.4]{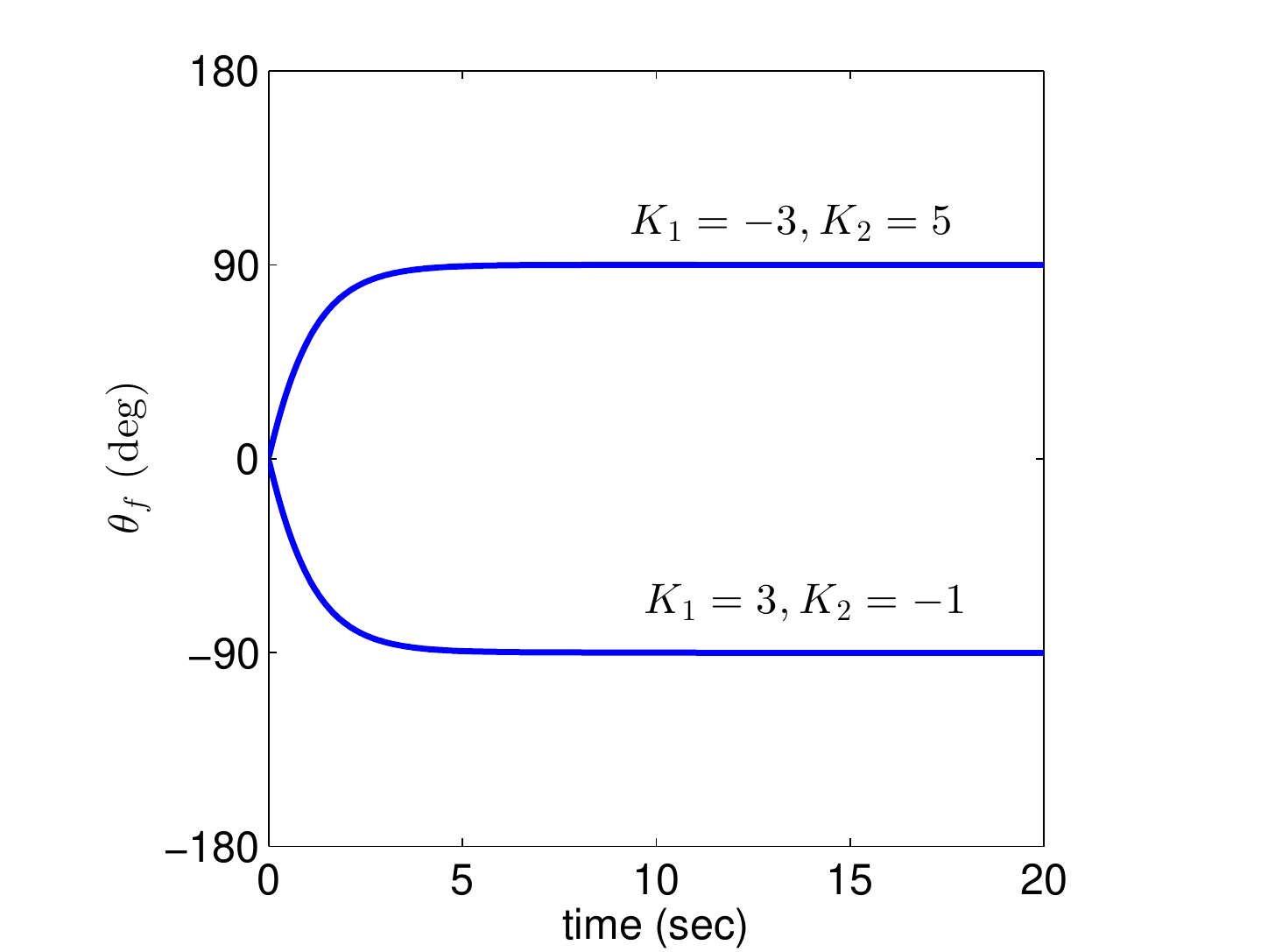}
\caption{Variation in the reference direction $\theta_f$ of the two agents with time for the heterogeneous gains $K_{set5} = \{3, -1\}$, and $K_{set6} = \{-3, 5\}$ under the control \eqref{control3_1}.}
\label{balanced formation two agents}
\end{figure}


\begin{example}\label{ex3}
Consider two agents starting from initial positions $\pmb{r}(0) = [(-1, -2), (5, -2)]^T$ with initial heading angles $\theta_{10} = 0^\circ$ and $\theta_{20} = 120^\circ$. For this setting, the convergence of reference direction $\theta_f$ at $-90^\circ$ and $90^\circ$ is shown in Fig.~\ref{balanced formation two agents} for the two sets of gains $K_{set5} = \{3, -1\}$, and $K_{set6} = \{-3, 5\}$, respectively. Here, since $\tilde{\theta}_{m0} = -60^\circ$ and $\tilde{\theta}_{M0} = 0^\circ$, only $\theta_f \in (-60^\circ, 0^\circ)$ would have been achievable for non-negative heterogeneous gains. However, by using a combination of both positive and negative heterogeneous gains $K_1$ and $K_2$ provided that $K_1 + K_2 > 0$, any $\theta_f \in [-180^\circ, 180^\circ]$ is reachable.
\end{example}

\section{Explicit Expressions of Velocity Directions And Convergence Point For Two Agents}
In this section, we try to obtain the explicit expressions of the velocity directions of the agents and their convergence point as a function of time in balanced formation. At first, we derive the explicit expressions of the velocity directions and then, by using these we obtain an explicit expression of the convergence point. We present the results only for $\omega_0 = 0$ since the analysis is unchanged for $\omega_0 \neq 0$ in a rotating frame of reference by redefining $\theta_k \rightarrow \theta_k + \omega_0 t$ for the $k^\text{th}$ agent.

\subsection{Velocity Directions}
For two agents, the explicit expressions of the velocity directions as a function of time are obtained as follows.

From \eqref{control3} and \eqref{angle_relation1}, one can form a differential equation for $N=2$ in terms of the heading angle $\theta_1(t)$ as
\begin{equation}
\label{explicit velocity direction}\dot{\theta}_1(t) + \frac{K_1}{2} \sin\left(K_2c_2 - \frac{1}{\lambda_2} \theta_1(t)\right)  = 0,
\end{equation}
where, $c_2 = \left(\theta_{10}/{K_1}\right) + \left(\theta_{20}/{K_2}\right)$, and $\lambda_2$ is already defined in \eqref{lamda1}.

Let
\begin{equation}
\label{relation12}K_2c_2 - \frac{1}{\lambda_2}\theta_1(t) = \delta(t),
\end{equation}
which, after differentiation with respect to time $t$, gives
\begin{equation}
\label{relation13} \dot{\theta_1}(t) = - \lambda_2\dot{\delta}(t).
\end{equation}
On substituting  $\dot{\theta}_1(t)$ from \eqref{relation13} and by using \eqref{relation12} in \eqref{explicit velocity direction}, we get a differential equation in terms of $\delta(t)$ as
\begin{equation}
\label{relation14}\dot{\delta}(t) - \kappa\sin\delta(t) = 0,
\end{equation}
where, $\kappa = \left(K_1 + K_2\right)/2 > 0$. Integrating both sides of \eqref{relation14} as following
\begin{equation}
\int_{\delta_0}^{\delta} \frac{d\delta}{\sin\delta} = \int_{0}^{t}\kappa dt,
\end{equation}
we get
\begin{equation}
\delta(t) = 2\tan^{-1}(\phi_0 e^{\kappa t}),
\end{equation}
where, $\delta_0 = \delta(0) = \theta_{20} - \theta_{10}$ (using \eqref{relation12}), and $\phi_0 = \tan\left(\delta_0/2\right)$.

Now, substituting for $\delta(t)$ in \eqref{relation12}, we get
\begin{equation}
\label{theta_1}\theta_1(t) =  \lambda_2\left\{K_2c_2 - 2\tan^{-1}(\phi_0 e^{\kappa t})\right\}.
\end{equation}
Also, substituting $\theta_1(t)$ in \eqref{angle_relation1} for $N=2$, we get
\begin{equation}
\label{theta_2}\theta_2(t) =  \lambda_1\left\{K_1c_2 + 2\tan^{-1}(\phi_0 e^{\kappa t})\right\},
\end{equation}
where, $\lambda_1$ is defined in \eqref{lamda1}.
These results show how heterogeneous controller gains affect the agents' velocity directions. Subtracting \eqref{theta_1} from \eqref{theta_2}, we get
\begin{equation}
\label{theta_2-theta_1}\theta_2(t) - \theta_1(t) =  2\tan^{-1}(\phi_0 e^{\kappa t}).
\end{equation}

In steady-state, that is, as $t \rightarrow \infty$, \eqref{theta_2-theta_1} simplified to
\begin{equation}
\theta_{2f} - \theta_{1f} =  \text{sgn}(\phi_0)\pi
\end{equation}
where, $\theta_{kf} = \theta_k (t \rightarrow \infty), k=1, 2$, and $\text{sgn}(\phi_0)$ is the signum function of $\phi_0$. Thus, the difference between the velocity directions of agents in phase balancing is, $|\theta_{2f} - \theta_{1f}| = \pi$ radians, as desired.

\subsection{Convergence Point}
The centroid of a group of agents is stabilized to a fixed point when they form a balanced formation. This fixed point is called the convergence point of the system. Thus, the convergence point is the centroid of the group as $t\rightarrow\infty$. It will be shown in this section that a desired convergence point can be achieved by suitably selecting the heterogeneous controller gains $K_1$ and $K_2$ of the two agents.

Let $x_c(t)$ and $y_c(t)$ are the abscissa and the ordinate of the centroid of the group at any time instant $t$. Then, the rate of change of centroid's position in \eqref{R_dot} can be written as
\begin{equation}
\label{R_dot_New}\dot{R} = \dot{x}_c+i\dot{y}_c = \frac{1}{N}\sum_{k=1}^{N} e^{i\theta_k}.
\end{equation}
For $N=2$, \eqref{R_dot_New} gives
\begin{equation}
\label{xc_dot}\dot{x}_c = \frac{1}{2}\left[\cos\theta_1+\cos\theta_2\right];~~~~\dot{y}_c = \frac{1}{2}\left[\sin\theta_1+\sin\theta_2\right].
\end{equation}
Integrating \eqref{xc_dot}, we get
\begin{eqnarray}
\label{xc_yc_after_integral}x_c (t)-x_{c0} = \frac{1}{2}\int_{0}^{t}\left\{\cos\theta_1+\cos\theta_2\right\}dt\\
\label{xc_yc_after_integral_1}y_c (t)-y_{c0} = \frac{1}{2}\int_{0}^{t}\left\{\sin\theta_1+\sin\theta_2\right\}dt,
\end{eqnarray}
where, $(x_{c0}, y_{c0}) = (x_c(0), y_c(0))$, denotes the coordinates of the initial location of the centroid.

We can compute the above integrals by using the following trigonometric relations:
\begin{eqnarray}
\label{trig_relation1}\cos\theta_1+\cos\theta_2 &=& 2 \cos\left(\frac{\theta_1+\theta_2}{2}\right)\cos\left(\frac{\theta_1-\theta_2}{2}\right)\\
\label{trig_relation2}\sin\theta_1+\sin\theta_2 &=& 2 \sin\left(\frac{\theta_1+\theta_2}{2}\right)\cos\left(\frac{\theta_1-\theta_2}{2}\right)
\end{eqnarray}
From \eqref{theta_1} and \eqref{theta_2}, we have
\begin{equation}
\label{relation15}\left(\theta_1+\theta_2\right)/{2} = \lambda_1\theta_{10} + \lambda_2\theta_{20} + (\lambda_1 - \lambda_2)\tan^{-1}(\phi_0 e^{\kappa t})
\end{equation}
and
\begin{equation}
\left(\theta_1-\theta_2\right)/2 =  \tan^{-1}(-\phi_0 e^{\kappa t}).
\end{equation}
Thus,
\begin{equation}
\label{relation16}\cos\left(\frac{\theta_1-\theta_2}{2}\right) = \cos\left(\tan^{-1}(-\phi_0 e^{\kappa t})\right) = \dfrac{1}{\sqrt{1 + \phi^2_0 e^{2\kappa t}}}.
\end{equation}
Using relations \eqref{trig_relation1} and \eqref{trig_relation2} along with \eqref{relation15} and \eqref{relation16} in \eqref{xc_yc_after_integral} and \eqref{xc_yc_after_integral_1}, we get
\begin{equation}
\label{xc_t_final}x_c (t)-x_{c0} = \int_{0}^{t} f(t)dt;~~~y_c (t)-y_{c0} = \int_{0}^{t} g(t)dt.
\end{equation}
where,
\begin{eqnarray}
\label{ft}{f}(t)=  \dfrac{\cos\left({\lambda_1\theta_{10} + \lambda_2\theta_{20} + (\lambda_1 - \lambda_2)\tan^{-1}(\phi_0 e^{\kappa t})}\right)}{\sqrt{1 + \phi^2_0 e^{2\kappa t}}} ,\\
\label{gt}{g}(t) =  \dfrac{\sin\left({\lambda_1\theta_{10} + \lambda_2\theta_{20} + (\lambda_1 - \lambda_2)\tan^{-1}(\phi_0 e^{\kappa t})}\right)}{\sqrt{1 + \phi^2_0 e^{2\kappa t}}} .
\end{eqnarray}

Above expressions provide the position of the centroid at any instant of time provided we are able to integrate these. Here, we are mainly interested to find out the steady-state position of the centroid, that is, the convergence point of the system. Thus, as $t \rightarrow \infty$, the co-ordinates of the centroid's position from \eqref{xc_t_final} are given by
\begin{eqnarray}
\label{equation_xc and yc}x_c (t\rightarrow\infty)-x_{c0} = I_1; ~~~~ y_c (t\rightarrow\infty)-y_{c0} = I_2
\end{eqnarray}
where,
\begin{equation}
\label{integral}I_1 = \int_{0}^{\infty} f(t)dt,~~~\text{and}~~~I_2 = \int_{0}^{\infty} g(t)dt,
\end{equation}
are improper integrals. It is difficult to integrate $I_1$ and $I_2$ by using usual integrating methods. But, we can prove convergence of $I_1$ and $I_2$ to ensure that the steady-state location of the centroid exists. To prove the convergence, we will utilize the following results from \cite{Apostol1991} and \cite{Shilov1996}.

\begin{thm}\label{Theorem6}
(Comparison test): Suppose $0 \leq f(t) \leq g(t)$ for all $t>a$. If $\int_{a}^{\infty} g(t) dt$ converges, then $\int_{a}^{\infty} f(t) dt$ converges.
\end{thm}

\begin{thm}\label{Theorem7}
If an improper integral $\int_{a}^{\infty} |f(t)| dt$ converges then $\int_{a}^{\infty} f(t) dt$ converges.
\end{thm}

Now, we prove the convergence of $I_1$ and $I_2$ in the following Lemma.

\begin{lem}\label{lem3}
For the functions $f(t)$ and $g(t)$, given by \eqref{ft}, and \eqref{gt}, respectively, the integrals $I_1$ and $I_2$, defined in \eqref{integral}, converge.
\end{lem}

\begin{proof}
Let us define a function
\begin{equation}
{h}(t) = \dfrac{1}{\sqrt{1 + \phi^2_0 e^{2\kappa t}}} > 0,~\forall t.
\end{equation}
Note that
\begin{eqnarray}
\label{ft_gt_ht}\left|{f}(t)\right| \leq {h}(t);~~~~\left|{g}(t)\right| \leq {h}(t).
\end{eqnarray}
Now, we define integral $I$ as
\begin{equation}
\label{integral_new}I = \lim_{t \to \infty} \int_{0}^{t} {h}(t) dt = \lim_{t \to \infty} \int_{0}^{t} \dfrac{dt}{\sqrt{1 + \phi^2_0 e^{2\kappa t}}}.
\end{equation}
Integrating \eqref{integral_new}, we get
\begin{equation}
\label{value}I = \frac{1}{2\kappa}\ln \left(\frac{\sqrt{1+\phi^2_0} + 1}{\sqrt{1+\phi^2_0} - 1}\right),
\end{equation}
which is finite except at $\phi_0 = 0$. Since, as defined above, $\phi_0 = \tan(\delta_0/2)$, where, $ \delta_0 = \theta_{20} - \theta_{10}$,  $\phi_0 = 0$, is a trivial case, as in this situation $\theta_{20} - \theta_{10} = 2n\pi, n \in \mathbb{Z}$, which says that the agents are initially in synchronized or in balanced formation.

For $\phi_0 \neq 0$, since the integral $\int_{0}^{\infty} {h}(t) dt$ converges to a finite value given by \eqref{value}, the integrals $\int_{0}^{\infty}|f(t)|dt$ and $\int_{0}^{\infty} |g(t)| dt$ converges as \eqref{ft_gt_ht} holds (Theorem~\ref{Theorem6}). Now, by using Theorem~\ref{Theorem7}, we conclude that $I_1$ and $I_2$ converges, and hence  $x_c(t\rightarrow\infty)$ and $y_c(t\rightarrow\infty)$ exist, that is, the centroid of group stabilizes to a fixed point (convergence point). This completes the proof.
\end{proof}

\subsection{Locus of Convergence Points}
In this subsection, we will find the locus of convergence points by varying the controller gains $K_1$ and $K_2$ in a way that the ratio ${K_1}/{K_2}$ is fixed. Since the analysis is quite involved, therefore, the assumption of fixing the ratio ${K_1}/{K_2}$ is made to carried out a few interesting results mentioned in the next theorem.

Let us assume
\begin{equation}
K_1 = \eta;~~~K_2 = \eta/\rho,
\end{equation}
where, $\rho$ is assumed to be constant. Thus the ratio $K_1/K_2 = \rho$ is fixed.

Since the phase balancing of two agents is achieved when the gains $K_1$ and $K_2$ satisfy $K_1 + K_2 > 0$ (Theorem~\ref{Theorem5}), it implies that $\eta\rho(\rho+1) > 0$ should hold here to ensure the same. The following conditions on $\eta$ and $\rho$ should fulfill to satisfy this inequality:
\begin{eqnarray}
\label{inequality} \eta\rho(\rho+1) > 0~~\Rightarrow~~
\begin{cases}
\eta < 0;~~-1 < \rho < 0\\
\eta > 0;~~\rho \in (-\infty, -1) \bigcup (0, \infty).
\end{cases}
\end{eqnarray}
Thus, for a given $\rho$, $\eta$ should be varied in such a way so that \eqref{inequality} is satisfied. Based on these notations, the following theorem is now stated.

\begin{thm}\label{Theorem8}
Consider two agents, with dynamics given by \eqref{modelNew}, under the control law \eqref{control3} with controller gains $K_1 = \eta$ and $K_2 = \eta/\rho$, where, $\eta$ and $\rho$ satisfy \eqref{inequality}. Let the initial heading angles of the agents be given by $[\theta_{10}, \theta_{20}]^T \in (-\pi, \pi)^2$. Then, in balanced formation of this system of two agents, the locus of the convergence point with different $\eta$ but fixed $\rho$, is a straight line approaching to the initial centroid $(x_{c0}, y_{c0})$ as $\eta \rightarrow \infty$.
\end{thm}

\begin{figure}
\centering
\includegraphics[scale=1]{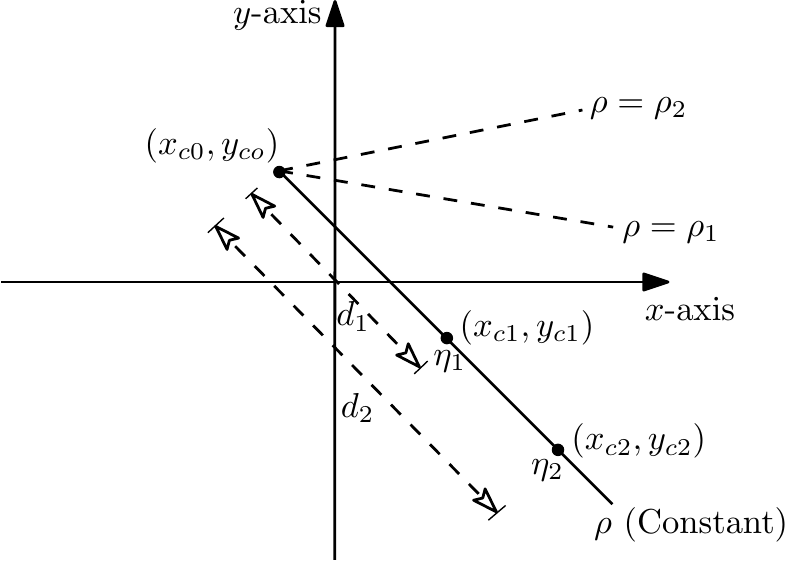}
\caption{Locus of the convergence point for fixed $\rho$. Corresponding to controller gains $K_1 = \eta_1$ and $K_2 = \eta_1/\rho$, the convergence point is $(x_{c1}, y_{c1})$, while for $K_1 = \eta_2$ and $K_2 = \eta_2/\rho$, the convergence point is $(x_{c2}, y_{c2})$. The point $(x_{c0}, y_{c0})$ is the initial centroid.}
\label{Locus of the convergence point for fixed}
\end{figure}

\begin{proof}
From \eqref{equation_xc and yc}, rewrite the coordinates of the convergence point as
\begin{eqnarray}
\label{xc_final}x_c (\infty)-x_{c0} = \int_{0}^{\infty}  f(t)dt\\
\label{y_c final}y_c (\infty)-y_{c0} = \int_{0}^{\infty} g(t)dt,
\end{eqnarray}
where, we denote $x_c (t\rightarrow\infty) = x_c (\infty)$ and $y_c (t\rightarrow\infty) = x_c (\infty)$, and the functions $f(t)$ and $g(t)$ are given by \eqref{ft} and \eqref{gt}, respectively. We can further simplify \eqref{xc_final} and \eqref{y_c final} as follows.
Let
\begin{eqnarray}
\label{relation17}\tan^{-1}(\phi_0 e^{\kappa t}) = \xi~~~\Rightarrow~~~dt = \frac{d\xi}{\kappa\sin\xi\cos\xi}
\end{eqnarray}
Substituting \eqref{relation17} in \eqref{xc_final} and \eqref{y_c final}, and accordingly changing the limits of the integrations, the coordinates of the convergence point are given by
\begin{eqnarray}
\label{xc_final_1}x_c (\infty)-x_{c0} = \frac{1}{\kappa}\int_{\frac{\delta_0}{2}}^{\frac{\pi}{2}} \dfrac{\cos\left({\lambda_1\theta_{10} + \lambda_2\theta_{20} + (\lambda_1 - \lambda_2)\xi}\right)}{\sin\xi}d\xi\\
\label{yc_final_1}y_c (\infty)-y_{c0} = \frac{1}{\kappa}\int_{\frac{\delta_0}{2}}^{\frac{\pi}{2}} \dfrac{\sin\left({\lambda_1\theta_{10} + \lambda_2\theta_{20} + (\lambda_1 - \lambda_2)\xi}\right)}{\sin\xi}d\xi
\end{eqnarray}

Now, substituting $K_1 = \eta$ and $K_2 = \eta/\rho$ in \eqref{xc_final_1} and \eqref{yc_final_1}, we get
\begin{eqnarray}
\label{xc_final_2}x_c (\infty)-x_{c0} = \frac{2\rho}{\eta(1+\rho)}\int_{\frac{\delta_0}{2}}^{\frac{\pi}{2}} \dfrac{\cos(f(\xi))}{\sin\xi}d\xi\\
\label{yc_final_2}y_c (\infty)-y_{c0} = \frac{2\rho}{\eta(1+\rho)}\int_{\frac{\delta_0}{2}}^{\frac{\pi}{2}} \dfrac{\sin(f(\xi))}{\sin\xi}d\xi.
\end{eqnarray}
where,
\begin{equation}
f(\xi) = \left[{\left(\frac{1}{1+\rho}\right)\theta_{10} + \left(\frac{\rho}{1+\rho}\right)\theta_{20} + \left(\frac{1-\rho}{1+\rho}\right)\xi}\right].
\end{equation}

Since the integrals $I_1$ and $I_2$ converge (Lemma~\ref{lem3}), the integrals in \eqref{xc_final_2} and \eqref{yc_final_2} also converge as these are obtained by change of variables in the original integrals $I_1$ and $I_2$. Moreover, since the above integrals are depended only on the constants $\rho$, and on given initial heading angles $\theta_{10}$ and $\theta_{20}$ (as $\delta_0 = \theta_{20} - \theta_{10}$), it may be assumed that these converge to constants say $h_1(\rho, \theta_{10}, \theta_{20})$ and $h_2(\rho, \theta_{10}, \theta_{20})$, respectively. Thus, we can write
\begin{eqnarray}
\label{xc_final_3}x_c (\infty)-x_{c0} &=& \frac{2\rho}{\eta(1+\rho)} h_1(\rho, \theta_{10}, \theta_{20})\\
\label{yc_final_3}y_c (\infty)-y_{c0} &=& \frac{2\rho}{\eta(1+\rho)} h_2(\rho, \theta_{10}, \theta_{20}).
\end{eqnarray}
Dividing \eqref{yc_final_3} by \eqref{xc_final_3}, we get
\begin{equation}
\frac{y_c (\infty)-y_{c0}}{x_c (\infty)-x_{c0}} = \frac{h_2(\rho, \theta_{10}, \theta_{20})}{h_1(\rho, \theta_{10}, \theta_{20})} = \overline{h}(\rho, \theta_{10}, \theta_{20})~\text{(say)},
\end{equation}
which is a constant (for fixed $\rho$, and given initial headings $\theta_{10}$, and $\theta_{20}$), and is independent of variable $\eta$. It implies that the locus of the convergence point for different values of $\eta$ is a straight line with slope $\overline{h}$ provided the ratio ${K_1}/{K_2} = \rho$ is fixed.

Also, as $\eta \rightarrow \infty$, we can get the coordinates of the convergence point from \eqref{xc_final_3} and \eqref{yc_final_3} as
\begin{eqnarray}
\label{xc_final_4}\lim_{\eta \rightarrow \infty} \left[x_c (\infty)-x_{c0}\right] &=& \lim_{\eta \rightarrow \infty}\frac{1}{\eta} \left(\frac{2\rho}{1+\rho}\right) h_1 = 0\\
\label{yc_final_4}\lim_{\eta \rightarrow \infty} \left[y_c (\infty)-y_{c0}\right] &=& \lim_{\eta \rightarrow \infty} \frac{1}{\eta}\left(\frac{2\rho}{1+\rho}\right) h_2 = 0,
\end{eqnarray}
which implies that
\begin{equation}
x_c(\infty) \rightarrow x_{c0};~~~y_c(\infty) \rightarrow y_{c0}~~~\text{as}~~\eta \rightarrow \infty.
\end{equation}
It means that the convergence point approaches the initial centroid for large value of $\eta$. This completes the proof.
\end{proof}

Pictorially, Theorem~\ref{Theorem8} is summarized in Fig.~\ref{Locus of the convergence point for fixed}. Now, we state the following corollaries to Theorem~7.

\begin{cor}
In Fig.~\ref{Locus of the convergence point for fixed}, let $(x_{c1}, y_{c1})$ and $(x_{c2}, y_{c2})$ be the locations of the convergence point for the gain pairs $(K_1, K_2) = \left(\eta_1, {\eta_1}/{\rho}\right)$ and $(\hat{K}_1, \hat{K}_2) = \left(\eta_2, {\eta_2}/{\rho}\right)$, respectively. Then, the relation
\begin{equation}
d_1\eta_1 = d_2\eta_2
\end{equation}
holds under the conditions given in Theorem~\ref{Theorem8}, where, $d_1$ and $d_2$ are the respective distances of the points $(x_{c1}, y_{c1})$ and $(x_{c2}, y_{c2})$ from the initial centroid $(x_{c0}, y_{c0})$.
\end{cor}

\begin{proof}
With reference to Fig.\ref{Locus of the convergence point for fixed}, we can write
\begin{eqnarray}
\label{d}d_{k} = \sqrt{(x_{ck}-x_{c0})^2+(y_{ck}-y_{c0})^2};~~~~k=1,2.
\end{eqnarray}
By using \eqref{xc_final_3} and \eqref{yc_final_3}, \eqref{d} can be written as
\begin{eqnarray}
\label{d_1}d_{k}= \frac{1}{\eta_{k}} \left(\frac{2\rho}{1+\rho}\right) \sqrt{h_1^2+h_2^2};~~~~k=1,2.
\end{eqnarray}
From \eqref{d_1}, we can conclude that $d_1\eta_1 = d_2\eta_2$. This result implies that we can select gain parameter $\eta_2$ to reach the new destination $(x_{c2}, y_{c2})$ on the same locus line with fixed $\rho$ if we have information about current gain $\eta_1$ and location $(x_{c1}, y_{c1})$. This completes the proof.
\end{proof}

\begin{cor}
For the conditions given in Theorem~\ref{Theorem8}, if $K_1 = K_2 = K>0$, then, the trajectories of both the agents, in balanced formation, are normal to the locus of convergence points.
\end{cor}

\begin{proof}
If $K_1 = K_2 = K > 0$, then $\lambda_1 = \lambda_2 = {1}/{2}$ and $\kappa = K$. Substituting these values in \eqref{xc_final} and \eqref{y_c final}, the coordinates of the convergence point are given by
\begin{eqnarray}
\label{xc_same_gain}x_c (\infty)-x_{c0} =  \cos\left(\frac{\theta_{10} + \theta_{20}}{2}\right) \int_{0}^{\infty} \dfrac{dt}{\sqrt{1 + \phi^2_0 e^{2Kt}}}\\
\label{yc_same_gain}y_c (\infty)-y_{c0} =  \sin\left(\frac{\theta_{10} + \theta_{20}}{2}\right) \int_{0}^{\infty} \dfrac{dt}{\sqrt{1 + \phi^2_0 e^{2Kt}}}.
\end{eqnarray}
Integrating \eqref{xc_same_gain} and \eqref{yc_same_gain}, we get
\begin{eqnarray}
x_c (\infty)-x_{c0} = \frac{1}{2K} \cos\left(\frac{\theta_{10} + \theta_{20}}{2}\right) \ln \left(\frac{\sqrt{1+\phi^2_0} + 1}{\sqrt{1+\phi^2_0} - 1}\right)\\
y_c (\infty)-y_{c0} = \frac{1}{2K} \sin\left(\frac{\theta_{10} + \theta_{20}}{2}\right) \ln \left(\frac{\sqrt{1+\phi^2_0} + 1}{\sqrt{1+\phi^2_0} - 1}\right),
\end{eqnarray}
which implies that
\begin{eqnarray}
\frac{y_c (\infty)-y_{c0}}{x_c (\infty)-x_{c0}} = \tan\left(\frac{\theta_{10} + \theta_{20}}{2}\right).
\end{eqnarray}
Therefore, for $K_1 = K_2 = K$, the locus of convergence point is a straight line of slope $m_1 = \tan\left(({\theta_{10} + \theta_{20}})/{2}\right)$.

Also, from \eqref{theta_1} and \eqref{theta_2}, it can be observed that the slopes of the straight line trajectories of both the agents in balanced formation, are same and is given by
\begin{equation}
m_2 = \tan\left(\theta_1\left(t\rightarrow\infty\right)\right) = \tan\left(\theta_2\left(t\rightarrow\infty\right)\right) = -\cot\left(\dfrac{\theta_{10}+\theta_{20}}{2}\right).
\end{equation}
Thus, $m_1m_2 = -1$, which is a condition when two straight lines of slopes $m_1$ and $m_2$ are perpendicular. Hence, this can be concluded that the trajectories of the agents are normal to the locus of the convergence point.
\end{proof}

\begin{remark}
Note that, for controller gains $K_1 = K_2 = K > 0$, the ratio $K_1/K_2 (=\rho)$ is fixed and unity. In this situation, the locus of the convergence is a straight line of slope $\tan((\theta_{10} + \theta_{20})/2)$. However, by using heterogeneous gains $K_1$ and $K_2$, we can get any desired convergence point in a two dimensional plane corresponding to different values of $\rho$ (see Fig.~\ref{Locus of the convergence point for fixed}). Therefore, by using heterogeneous controller gains, in balanced formation, we can regulate the velocity directions as well as convergence point of the agents in balanced formation.
\end{remark}


\begin{figure}[!t]
\centering
\includegraphics[scale=0.4]{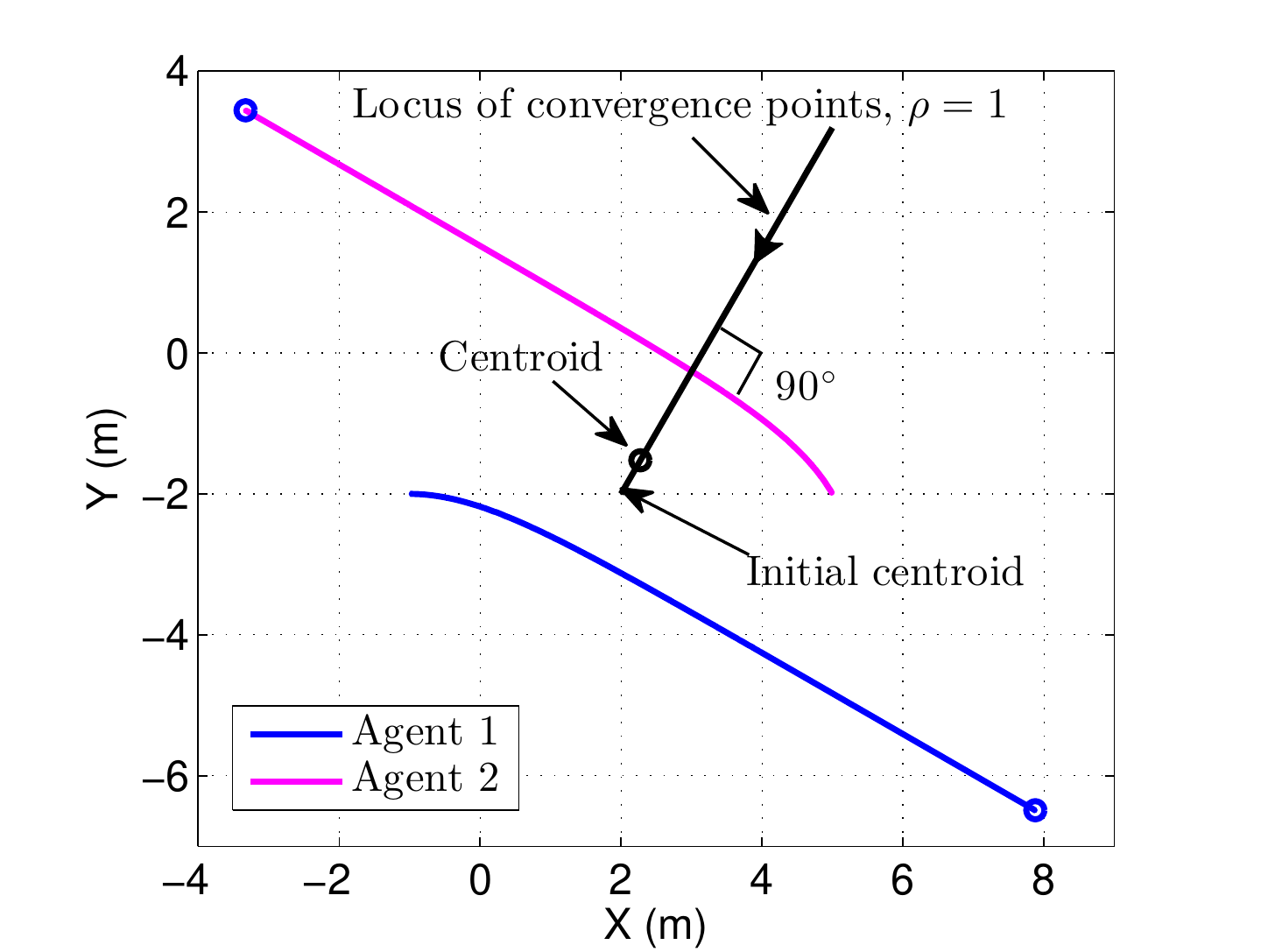}
\caption{The locus of the convergence points for $\rho = 1$, and the trajectories of the agents under the control \eqref{control3_1} with $K_1 = K_2 =1$. The trajectories of the agents are normal to the locus of the convergence points.}
\label{locus of convergence points}
\end{figure}


\begin{example}
Two agents are considered as in Example~\ref{ex3}. Note that the initial centroid is located at $(x_{c0}, y_{c0}) = (2, -2)$. The locus of the convergence points for various values of the heterogeneous gains $K_1$ and $K_2$ provided $\rho = 1$, is also shown in Fig.~\ref{locus of convergence points}, which is straight line of slope $\sqrt{3}$, and passes through the the point $(x_{c0}, y_{c0})$. The trajectories of the agents in balanced formation are also shown in the same figure for the heterogeneous gains $K_1 = K_2 = 1$. Clearly, trajectories of the agents are perpendicular to the locus of the convergence points.
\end{example}

\begin{remark}
In practical, autonomous vehicle can withstand a limited control force due to physical constraints. In such a case, the concept of heterogeneous controller gains can be used to ensure that the applied control force does not violates the maximum allowable limit. This situation is already addressed in \cite{Jain2016} in achieving synchronization, and can be equivalently stated in phase balancing.
\end{remark}

\section{Conclusions}
In this paper, we have investigated the phenomenon of phase balancing in a group of heterogeneously coupled agents. It has been shown that a desired reference direction, and hence, the desired orientations of the agents in balanced formation, can be achieved by appropriately selecting the heterogeneous controller gains $K_k, \forall k$, given according to the Assumption~\ref{assumption}. Moreover, it has been illustrated through simulation that the reachable set of the reference direction further expands when both positive and negative values of the heterogeneous gains are incorporated in the control scheme. In particular, it has been proved analytically for $N=2$ that there exists a condition on the heterogeneous controller gains which allows them to assume both positive and negative values, and hence, results in, further expansion of the reachable set of the reference direction. By obtaining the closed form expressions of the velocity directions for $N=2$, we have further shown that the locus of the convergence point, for various values of the heterogenous gains provided their ratio is fixed, is a straight line passing through the initial centroid. Furthermore, it has been pointed out for realistic systems that an upper bound on the control force, applied to each agent, can be obtained by bounding the heterogeneous control gains.

Simulation results show the effectiveness of using heterogenous control gains in regulating the velocity directions of $N$ agents in balanced formation. It would be interesting as a future research to find out an analytical expression relating angular separation between agents' velocity vectors in balanced formation with the heterogeneous control gains for the general case of $N$ agents. The consideration of issue of collision avoidance among agents is also an interesting future problem.

\section{Appendix}
The splay phase (a special case of phase balancing) is an arrangement in which the agents are at equal angular separation, that is, their phases are separated by multiples of $2\pi/N$. The $m^\text{th}$ harmonic of the phase order parameter $p_\theta$, which plays an important role in stabilizing the splay phase arrangement, is given by \cite{Paley2005}
\begin{equation}
\label{m_phase_order_parameter}p_{m\theta} = \frac{1}{mN}\sum_{k=1}^{N}e^{im\theta_k} = |p_{m\theta}|e^{i\Psi_m},
\end{equation}
where, $m \in \mathbb{N} \triangleq \left\{1, 2, 3, \ldots\right\}$, and $0 \leq \left|p_{m\theta}\right| \leq 1/m$. The splay phase arrangement occurs when the condition
\begin{equation}
\label{condition}p_{1\theta} = p_{2\theta} = \ldots = p_{\lfloor N/2 \rfloor \theta} = 0
\end{equation}
holds \cite{Sepulchre2007}. Condition \eqref{condition} indicates that the splay phase arrangement corresponds to the phase balancing of the first $\lfloor N/2 \rfloor$ harmonics of $p_\theta$. Therefore, in order to stabilize the splay phase arrangements, we use the potential function given as,
\begin{equation}
\label{potential function_splay}W(\pmb{\theta}) = \frac{N}{2}\sum_{m=1}^{\lfloor N/2 \rfloor}|p_{m\theta}|^2
\end{equation}
which is minimized in the splay formation. Also, $|p_{1\theta}| = |p_\theta| = 0$ corresponds to the general phase balancing as discussed above. For the sake of convenience, let us define
\begin{equation}
\label{m_potential function_new}U_m(\pmb{\theta}) = \frac{N}{2}|p_{m\theta}|^2,
\end{equation}
by using which \eqref{potential function_splay} can be rewritten as
\begin{equation}
\label{potential function_new}W(\pmb{\theta}) = \sum_{m=1}^{\lfloor N/2 \rfloor} U_m(\pmb{\theta}).
\end{equation}
A Lyapunov-based control framework exists to stabilize splay formation as discussed in the following theorem.

\begin{thm}\label{Theorem1_splay}
Consider the system dynamics \eqref{modelNew} with control law
\begin{equation}
\label{control1_splay}u_k = -K_k\left(\frac{\partial W}{\partial \theta_k}\right);~~~K_k \neq 0,
\end{equation}
and define a term
\begin{equation}
\label{term_splay}\overline{T}_k(\pmb{\theta}) = \left(\frac{\partial W}{\partial \theta_k}\right)^2
\end{equation}
for all $k = 1, \ldots, N$, where $m = 1, \ldots, {\lfloor N/2 \rfloor}$. If $\sum_{k=1}^{N} K_k T_k(\pmb{\theta}) > 0$, all the agents asymptotically stabilize to a splay formation. Moreover, $K_k > 0, \forall k$, is a restricted sufficient condition in stabilizing splay formation.
\end{thm}

\begin{proof}
The proof proceeds on the same steps as used to prove Theorem~\ref{Theorem1}. We just need to analyze the invariant set where $\dot{W}(\pmb{\theta}) = 0$, and the critical points of $W(\pmb{\theta})$.

Since $\pmb{\theta} \in \mathbb{T}^N$ is compact, it follows from LaSalle's invariance theorem \cite{Khalil2000} that all the solutions of \eqref{modelNew} under control \eqref{control1_splay} converge to the largest invariant set contained in $\{\dot{W}(\pmb{\theta}) = 0\}$, that is, the set
\begin{equation}
\Omega = \left\{\pmb{\theta}~|~ ({\partial W}/{\partial \theta_k}) = \sum_{m=1}^{\lfloor N/2 \rfloor}\left<p_{m\theta}, ie^{im\theta_k}\right> = 0,~\forall k\right\},
\end{equation}
which is also the critical set of $W(\pmb{\theta})$. In this set, dynamics \eqref{modelNew2} reduces to $\dot{\theta}_k = 0, \forall k$, which implies that all the agents move in a straight line. The set $\Omega$ is itself invariant since
\begin{align}
\nonumber\sum_{m=1}^{\lfloor N/2 \rfloor}\frac{d}{dt}\left<p_{m\theta}, ie^{im\theta_k}\right> &= -\sum_{m=1}^{\lfloor N/2 \rfloor}m\left<p_{m\theta}, e^{im\theta_k}\dot{\theta}_k\right> \\
\nonumber &+ \frac{1}{N}\sum_{m=1}^{\lfloor N/2 \rfloor}\left<\sum_{k=1}^{N} ie^{im\theta_k} \dot{\theta}_k, ie^{im\theta_k}\right>\\
\nonumber = -\sum_{m=1}^{\lfloor N/2 \rfloor}m\left<p_{m\theta}, e^{im\theta_k}\right>\dot{\theta}_k  & + \sum_{m=1}^{\lfloor N/2 \rfloor}m\left<p_{m\theta}, e^{im\theta_k}\right>\dot{\theta}_k = 0
\end{align}
on this set. Therefore, all the trajectories of the system \eqref{modelNew} under control \eqref{control1} asymptotically converges to the critical set of  $W(\pmb{\theta})$.

{\it Analysis of the critical points}:

The critical points of $W(\pmb{\theta})$ are given by the $N$ algebraic equations
\begin{equation}
\frac{\partial W}{\partial \theta_k} = \sum_{m=1}^{\lfloor N/2 \rfloor}\left<p_{m\theta}, ie^{im\theta_k}\right> =  \sum_{m=1}^{\lfloor N/2 \rfloor}|p_{m\theta}|\sin(\Psi_m - m\theta_k) = 0,~~1\leq k \leq N.
\end{equation}
Since the critical points with $p_{m\theta} = 0$, where $m = 1, \ldots, {\lfloor N/2 \rfloor}$, are the global minima of $W(\pmb{\theta})$, the splay phase arrangement is asymptotically stable if $K_k > 0,\forall k$. The rest of the critical points where $p_{m\theta} \neq 0$, and $\sin(\Psi_m - m\theta_k) = 0, \forall k$, are unstable points, the proof of which directly follows from the Theorem~2 in \cite{Paley2005} since the critical points are independent of the heterogeneous control gains. This completes the proof.
\end{proof}

The control law \eqref{control1_splay}, after simplification, can be written as
\begin{equation}
\label{control1_splay_new}\dot{\theta_k} = - \frac{K_k}{N} \sum_{j=1}^{N}\sum_{m=1}^{\lfloor N/2 \rfloor} \frac{1}{m} \sin(m(\theta_j - \theta_k)).
\end{equation}

From \eqref{m_phase_order_parameter}, we can write
\begin{equation}
\left|p_{m\theta}\right|e^{i(\Psi_m - m\theta_k)} =  \frac{1}{mN}\sum_{j=1}^{N}e^{im(\theta_j - \theta_k)},
\end{equation}
the imaginary part of which is given by
\begin{equation}
\label{phase_order_parameter_New_splay}\left|p_{m\theta}\right|\sin(\Psi_m - m\theta_k) = \frac{1}{mN}\sum_{j=1}^{N} \sin(m(\theta_j - \theta_k)).
\end{equation}
Using \eqref{phase_order_parameter_New_splay}, \eqref{control1_splay_new} can be written as
\begin{equation}
\label{theta_dot_new_splay}\dot{\theta}_k = -K_k\sum_{m=1}^{\lfloor N/2 \rfloor}\left|p_{m\theta}\right|\sin(\Psi_m - m\theta_k),
\end{equation}
which, for $m =1, 2$, and $3$, results in the same control as defined in \eqref{theta_dot_new}. However, for $m>3$, unlike \eqref{theta_dot_new}, in this case it may not be easy to speculate the result like Lemma~1 since \eqref{theta_dot_new_splay} contains $m$ harmonic terms, and hence, is a challenging problem.


\begin{figure}[!t]
\centering
\includegraphics[scale=0.4]{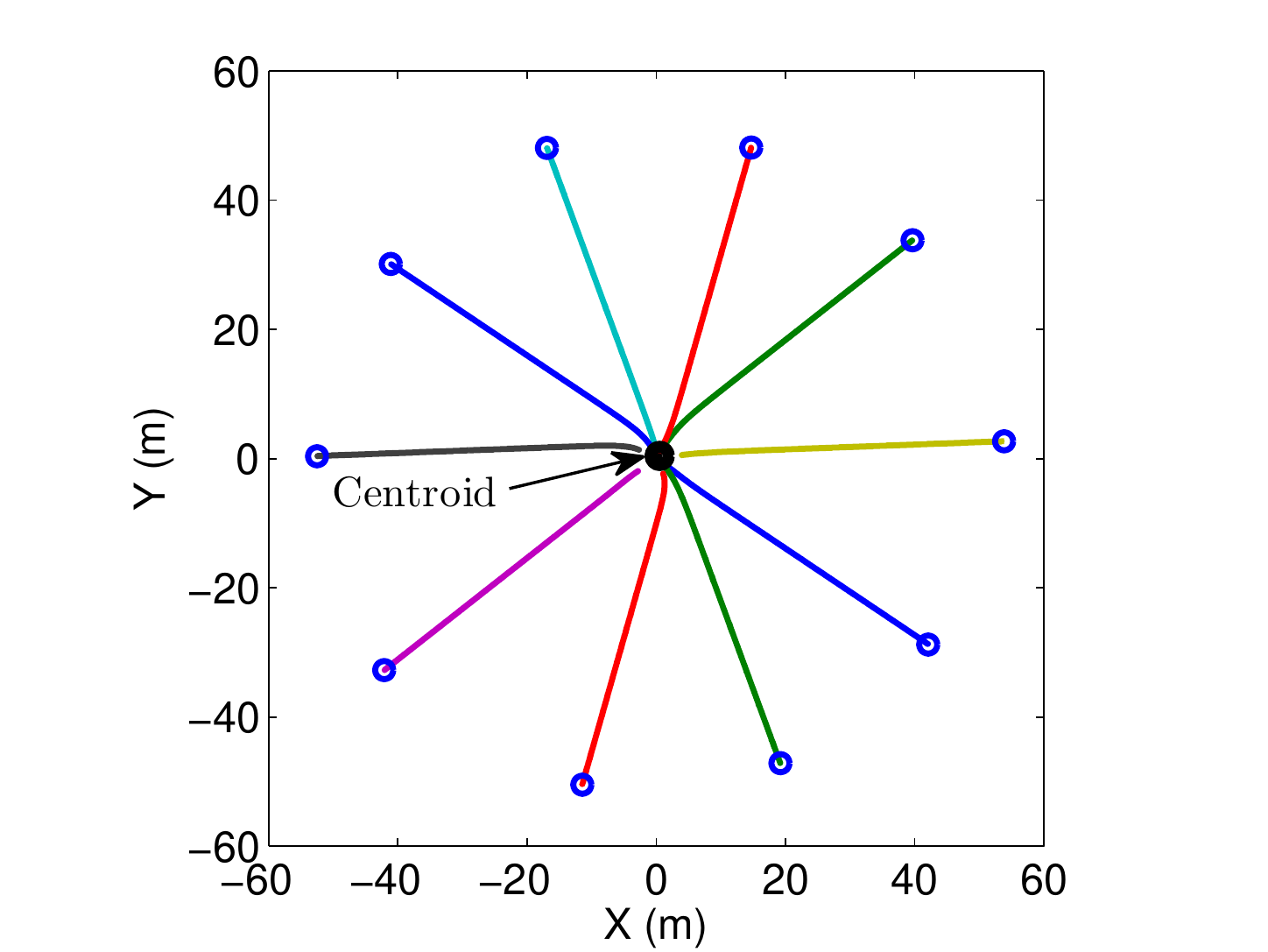}
\caption{Splay formation of ten agents under the control law \eqref{control1_splay_new} with heterogeneous gains $K_{set7} = {k, k = 1, \ldots, N}$.}
\label{splay formation}
\end{figure}


\begin{example}
The splay formation of the ten agents with arbitrary generated initial positions as well as the heading angles is shown in Fig.~\ref{splay formation} under the control \eqref{control1_splay_new} with heterogeneous gains $K_{set7} = {k, k = 1, \ldots, N}$. The angular separation between the velocity vectors of the consecutive agents is $36^\circ$ as desired.
\end{example}



\end{document}